\documentclass[journal]{IEEEtran}

\usepackage{amsmath,amssymb,amsthm,graphicx,mathrsfs,url,mathtools}
\usepackage[usenames,dvipsnames]{color}
\usepackage[hidelinks]{hyperref}
\usepackage{amsxtra}
\usepackage{epstopdf}

\let\varepsilonilon=\varepsilon 

\newcommand{\RR}{{\mathbb R}}

\newcommand{\CC}{{\mathbb C}}

\newcommand{\eps}{\varepsilon}

\newtheorem{theo}{Theorem}
\newtheorem{prop}{Proposition}[section]	
\newtheorem{defi}[prop]{Definition}

\newtheorem{lemm}[prop]{Lemma}
\newtheorem{corr}[prop]{Corollary}
\newtheorem{rem}{Remark}

\numberwithin{equation}{section}

\DeclareMathOperator{\Spec}{Spec}

\DeclareMathOperator{\tr}{Tr}

\def\indic{\operatorname{1\hskip-2.75pt\relax l}}

\newcommand{\eqal}[1]{\begin{equation}
\begin{aligned}
#1
\end{aligned}
\end{equation}}

\newcommand{\pr}[3]{H(#1|#2)_{#3}}

\newcommand{\hr}[2]{H(#1)_{#2}}

\newcommand{\ket}[1]{\vert #1 \rangle}

\newcommand{\kb}[2]{\left| #1 \vphantom{#2} \right>\left< #2 \vphantom{#1} \right|} 
\newcommand{\proj}[1]{\kb{#1}{#1}} 

\newcommand{\Tr}{\mathrm{Tr}}

\newcommand{\ha}{\hat{H}} 
\newcommand{\hb}{\hat{H}_+} 
\newcommand{\td}{\varepsilon} 

\newcommand{\cH}{\mathcal{H}}

\newcommand{\cA}{\mathbb{N}_0}

\newcommand{\con}[1]{\mathbb{E}\left[f(#1)\right]}
\newcommand{\cont}[1]{\mathbb{E}\left[\tilde{f}(#1)\right]}

\begin{document}

\title{From Classical to Quantum: Uniform Continuity bounds on entropies in Infinite Dimensions}

\author{Simon Becker,~Nilanjana~Datta~and~Michael~G~Jabbour
\thanks{S. Becker is in the Mathematics Department at ETH Zurich (e-mail: simon.becker@math.ethz.ch).

N. Datta is in the Department of Applied Mathematics and Theoretical Physics, Centre for Mathematical Sciences, University of Cambridge, Cambridge CB3 0WA, United Kingdom (e-mail: n.datta@statslab.cam.ac.uk).

M. G. Jabbour is in the Department of Physics, Technical University of Denmark, 2800 Kongens Lyngby, Denmark (e-mail: mgija@dtu.dk).}}

\maketitle

\begin{abstract}
	We prove a variety of improved uniform continuity bounds for entropies of both classical random variables on an infinite state space and of quantum states of infinite-dimensional systems. We obtain the first tight continuity estimate on the Shannon entropy of random variables with a countably infinite alphabet. The proof relies on a new mean-constrained Fano-type inequality. We then employ this classical result to derive a tight energy-constrained continuity bound for the von Neumann entropy.
	To deal with more general entropies in infinite dimensions, \textit{e.g.}~$\alpha$-R\'enyi and $\alpha$-Tsallis entropies, we develop a novel approximation scheme based on operator Hölder continuity estimates.
	Finally, we settle an open problem raised by Shirokov~\cite{S21,S21b} regarding the characterisation of states with finite entropy.
\end{abstract}

\IEEEpeerreviewmaketitle

\section{Introduction}

\IEEEPARstart{I}{t} has been known, at least since the thorough study of entropies in infinite dimensions in \cite{W78}, that the von Neumann entropy is a discontinuous function of density operators with respect to trace distance in any infinite-dimensional Hilbert space. In the same article, it has been observed that continuity can be restored by imposing an additional energy constraint on the density operators. Since then continuity bounds under energy constraints for entropies in infinite dimensions have been widely studied and a particularly comprehensive and practicable list of continuity bounds has been obtained by Winter \cite{W15}. Most noteworthy for our purposes is his Lemma $15$ in which he shows that for density operators satisfying a uniform energy constraint with respect to a Hamiltonian satisfying the so-called \emph{Gibbs hypothesis} (cf.~Def.~\ref{defi:Gibbs}), the von Neumann entropy becomes a continuous function of the density operator. Similar bounds for general $\alpha$-R\'enyi entropies or $\alpha$-Tsallis entropies, in the regime $\alpha\in (0,1)$, are missing and existing continuity bounds are limited to the finite-dimensional case \cite{K,HD}. We also provide estimates in the technically simpler regime for $\alpha>1.$ 

While the derivation of continuity bounds for classical entropies of random variables with infinite state space is interesting in itself, our ultimate goal is to provide new continuity bounds for the entropies of states of infinite-dimensional quantum systems. This would prove particularly useful in the context of continuous-variable (CV) quantum systems 
\textit{e.g.} collections of electromagnetic modes travelling along an optical fibre \cite{Serafini,Weedbrook2011}. The natural Hamiltonian of these systems is the so-called boson {\em{number operator}}. Such systems are of immense technological and experimental relevance since promising proposed protocols for quantum communication and computation rely on them. Consequently they have been the subject of extensive research in recent years. Continuity bounds for quantum entropies of states of CV systems which satisfy an energy constraint with respect to the number operator are of particular importance since they would lead to bounds on optimal communication rates in protocols which employ them.~\\

\noindent \textbf{Shannon and von Neumann entropy.} 
Our first result is concerned with providing a tight version of Winter's bound in \cite{W15} on the difference of von Neumann entropies for two states, when the Hamiltonian imposing an energy constraint on the input states is the number operator. The bound obtained by Winter is asymptotically tight,
see also \cite{Shirokov1,Shirokov2,SBND}. 

In contrast, for any given energy threshold $E$, our bound (cf. Theorem~\ref{theo:ContinuityQuantum}) is tight for all values of the trace distances ($\eps$) such that $\eps \in [0, E/(E+1)]$
. The proof of this new estimate builds upon a new Fano inequality for classical random variables on the natural numbers.
Let $X,Y$ be random variables with finite state space. Fano's inequality relates the conditional Shannon entropy $H(X|Y)$ and the error probability $\mathbb{P}(X \neq Y)$, and is one of the most elementary and yet 
important examples of entropic inequalities which are of fundamental importance in information theory.
If $X$ and $Y$ take values in a finite alphabet $\mathcal{A}$, and $\mathbb{P}(X \neq Y)= \varepsilon$, then Fano's inequality states that
\begin{equation}
    H(X|Y) \leq \varepsilon \log (|\mathcal{A}|- 1) + h(\varepsilon),
\end{equation}
where $h(\varepsilon) \coloneqq -\varepsilon \log \varepsilon - (1-\varepsilon) \log (1-\varepsilon)$ is the binary entropy. Throughout the manuscript $\log$ denotes the natural logarithm.

However, if the alphabet $\mathcal{A}$ is countably infinite, then the above inequality does not hold. In fact, it is even possible
for $H(X|Y)$ to remain non-zero in the limit $\varepsilon \to 0$, see \cite[Ex.2.49]{yeung2008information}. This is a consequence of the discontinuity of the Shannon entropy
for such alphabets. Hence, it is interesting to derive generalized forms of Fano's inequality in this case, 
under suitable constraints, which ensure that $H(X|Y)$ tends  to zero as $\varepsilon \to 0$. One such generalization was obtained by Ho and Verdu~\cite{HoVerdu2010}
in which the marginal $p_X$ of the joint distribution $p_{XY}$ was fixed. See also~\cite{Sakai} for further generalizations.
In this paper we consider a different generalization of Fano's inequality for a countably infinite alphabet, namely, one in which the means of $X$ and $Y$ are 
constrained to be below a prescribed threshold value. Hence, we refer to this inequality as a {\em{mean-constrained Fano's inequality}}.

We then employ this inequality to derive a tight uniform continuity bound for the Shannon entropy of random variables with a countably infinite alphabet by using the notion of maximal coupling, thus extending the work of Sason in~\cite{Sason2013}. Next we use this classical continuity bound to obtain a tight uniform continuity bound on the von Neumann entropy of states of an infinite-dimensional quantum system satisfying an energy constraint, in the case in which the Hamiltonian is the number operator.~\\

\begin{rem}
    After the publication of this article, Maksim Shirokov informed us that in Theorems~\ref{theo:FanoG},~\ref{theo:ContinuityG}, and~\ref{theo:ContinuityQuantumG}, the quantity $E$ needs to be chosen above a sufficiently small non-zero threshold. In a subsequent joint article \cite{BDJS_24}, we obtained sharp estimates for the full energy range and arbitrary trace distances between states, with analogous results in the classical setting, extending the estimates obtained in this work.
\end{rem}

\noindent \textbf{R\'enyi and Tsallis entropies.} We then turn to $\alpha$-R\'enyi- and $\alpha$-Tsallis entropies, for which in infinite dimensional quantum systems\footnote{For finite-dimensional systems, a tight uniform continuity bound for $\alpha$-R\'enyi entropies for $\alpha \in (0,1)$ was proved by Audenaert in~\cite{Audenaert}. } no general continuity bounds have been established so far when $\alpha \in (0,1).$ We also discuss the case that $\alpha>1$ where dimension-independent estimates have only been established for the Tsallis entropy \cite{tsallis,Audenaert,HD}. The corresponding bounds for the $\alpha$-R\'enyi entropy \cite{Audenaert,Chen_2016,HD} are dimension-dependent and therefore do not apply to infinite dimensional quantum systems. 
We would like to emphasize that the study of continuity bounds for $\alpha$-R\'enyi- and $\alpha$-Tsallis entropies is rather different from the von Neumann setting. For example, while the von Neumann entropy is a continuous map on the set of states with uniformly bounded energy with respect to the number operator, this fails to be true for general $\alpha$-R\'enyi and $\alpha$-Tsallis entropies. To illustrate this, consider a state $\rho$ with eigenvalues 
\[ \Spec(\rho)=\left\{ \frac{1}{(i+1)^{1/\alpha}} \frac{1}{\zeta(1/\alpha)} ; i \in \mathbb N_0\right\},\]
where $\zeta $ is the Riemann zeta function.
For $\alpha <1/2$, this state has bounded energy with respect to the number operator $\hat{N}$, as $\tr(\rho \hat{N}) <\infty$ for $\alpha <1/2.$ However, for this range of $\alpha$, $
\tr(\rho^{\alpha})=\infty$ which implies that neither the $\alpha$-R\'enyi entropy nor the $\alpha$-Tsallis entropy exists. Therefore, simple energy constraints by the number operator are in general insufficient to obtain continuity bounds for such entropies and we must use a different approach, as the one we pursue for the von Neumann entropy. Using recent results on operator H\"older continuous functions, we are able to obtain continuity estimates for $\alpha$-R\'enyi and Tsallis entropies. In fact, we identify technical spectral conditions on the Hamiltonian under which an energy constraint by the Hamiltonian gives rise to continuity bounds for any $\alpha \in (0,1)$ in Lemma \ref{lemm:moment_bounds} in the appendix. 

In this article, we obtain such bounds, after first deriving them for discrete and continuous random variables in Section \ref{sec:classicalent}. We do this by outlining a simple procedure in Theorem~\ref{theo:approx_theo}, based on the proof of the gentle measurement lemma \cite[Lemma $9$]{W99}, to obtain continuity bounds for H\"older continuous functions, but also for functions $f$ with different regularity, of density operators in infinite-dimensional Hilbert spaces under more refined energy constraints on the state. For instance, $\alpha$-R\'enyi  and $\alpha$-Tsallis entropies, for $\alpha \in (0,1)$, depend on functions $f_{\alpha}(x)=x^{\alpha}$ of the density operator. For states $\rho,\sigma$, our procedure allows us to obtain bounds on the trace distance $\Vert f_{\alpha}(\rho)-f_{\alpha}(\sigma)\Vert_1$ which easily leads to continuity bounds for entropies. In a nutshell, we identify constraints, cf. Theorem \ref{theo:approx_theo}, under which we can approximate the operator $f(\rho)$ by a finite rank operator. We then utilize a continuity result on the level of the finite-rank operator to derive a continuity estimate for the function of the density operator.~\\

\noindent \textbf{FA-property.} In a recent series of papers \cite{S21,S21b}, Shirokov identified a property that he coined the {\emph{Finite-dimensional Approximation (FA) property}} on density operators. He showed that the set of density operators satisfying FA contains almost all states of finite entropy but left open the following question: {\em{Does any state with finite entropy necessarily satisfy the FA-property?}} We give a negative answer to this question. The strength of the newly introduced FA-property is due to various approximation, continuity, and stability estimates obtained for states satisfying that property in the papers \cite{S21,S21b}. The FA property for a state is equivalent, see \cite[Theo.$1$]{S21} to the existence of a positive semi-definite Hamiltonian $\ha$ with discrete spectrum such that the state has finite energy $\tr(\rho \hat H)<\infty$. In addition, one requires $e^{-\beta \hat H}$ to be a trace-class operator with controlled limiting behaviour $\lim_{\beta \downarrow 0}\left(\tr(e^{-\beta \hat H})\right)^{\beta}=1.$ Using the work by Wehrl \cite{W78}, one could already conclude from this that any such state has finite entropy, but the converse implication was left as an open problem. We address this is open problem in Theorem \ref{theo:FAprop}.~\\

\noindent \textbf{A summary of our main results and layout of the paper.}
Note that through all this paper we consider random variables with infinite state spaces and quantum states of infinite-dimensional systems. 
\begin{itemize}
    \item{\emph{Shannon entropies}:} In Section \ref{sec:Shannon} we consider classical random variables with state space $\mathbb N_0$ and state a tight mean-constrained Fano's inequality in Theorem~\ref{theo:Fanob}. Employing this, we obtain a tight mean-constrained continuity bound for the Shannon entropy of such random variables. This is stated in Theorem~\ref{theo:Continuityb}. 
    \item{\emph{von Neumann entropies}:} In Section \ref{sec:ContinuityQuantum}, we obtain a tight energy-constrained continuity bound for von Neumann entropies of states of infinite-dimensional quantum systems, when the Hamiltonian is the number operator. This is stated in Theorem~\ref{theo:ContinuityQuantum}.
    \item{\emph{Classical $\alpha$-R\'enyi and $\alpha$-Tsallis entropies}:} In Section \ref{sec:classicalent} we derive continuity estimates for the classical  $\alpha$-R\'enyi and $\alpha$-Tsallis entropies for random variables with discrete and continuous state spaces in Corollaries \ref{corr:contdisc} and \ref{corr:contcont}, respectively, without any restrictions on the state space.
    \item{\emph{Quantum $\alpha$-R\'enyi and $\alpha$-Tsallis entropies:}} In Section \ref{sec:QEntropies} we then introduce a general approximation scheme for functions of quantum states in Theorem~\ref{theo:approx_theo} that allows us to obtain continuity estimates for the $\alpha$-R\'enyi and $\alpha$-Tsallis entropies with $\alpha \in (0,1)$ in Corollary \ref{corr:quantcont}. Analogous continuity bounds for the range $\alpha>1$ are given in Proposition~\ref{prop:alphag1}.
    \item{\emph{FA-property}:} The final section is on the FA property. In Theorem~\ref{theo:FAprop} we answer a question raised by Shirokov in \cite{S21,S21b} by showing that there exist states of finite entropy that do not satisfy the FA-property.  
\end{itemize}

\section{Mathematical preliminaries}
\noindent \textbf{Notation}
The countably infinite set of non-negative integers is denoted by $\mathbb N_0$ and the set of strictly positive ones by $\mathbb N$.
 We consider random variables $X,Y$ on some probability space $(\Omega,\Sigma, \mathbb P)$ and, if they are integrable, denote their expectation by $\mathbb E(X).$
For $X,Y$ taking values in a countably infinite state space (also referred to as alphabet) $Z \coloneqq \bigcup_{i \in \mathbb N} \{ z_i \}$ and positive weights $w=(w_i)$, the total variation distance is given by by
\begin{equation}
\begin{aligned}
\Vert X-Y\Vert_{\operatorname{TV}(w)}  & \coloneqq \operatorname{TV}_{(w)}(X,Y) \\
& \coloneqq  \frac{1}{2} \sum_{i \in \mathbb N}w_i \vert \mathbb P(X=z_i) - \mathbb P(Y=z_i)\vert.\label{TV}
\end{aligned}
\end{equation}
The spaces of $p$-summable and $p$-integrable functions with weights $w=(w_i)$ or $w(x)$ are denoted by $\ell^p(w)$ and $L^p(w)$ as usual. In the unweighted case, we just omit the argument $(w).$ 
We denote by $\indic_A$ the indicator function of a set $A.$

We denote by $\mathcal H$ a separable infinite-dimensional Hilbert space. The $p^{th}$ Schatten class on $\mathcal H$ is denoted by $\mathcal S_p$ with norm $\Vert \cdot \Vert_p$, and the operator norm is denoted by $\Vert \cdot \Vert$. In particular, $\mathcal{T}(\mathcal H) \coloneqq \mathcal S_1$ is the Banach space of trace class operators. A state (or density operator) $\rho$ is a positive trace-class operator of unit trace on $\mathcal H$. The trace is denoted by $\tr.$ The spectrum of a linear operator $T$ is denoted by $\Spec(T).$
Let $\rho,\sigma$ be states, the fidelity between them is defined as 
$F(\rho,\sigma) = \Vert \sqrt{\rho} \sqrt{\sigma}\Vert_1$.
The Fuchs-van de Graaf inequality then states that 

\begin{equation}
\label{eq:FvdG}
    1-F(\rho,\sigma) \le \frac{1}{2}\Vert \rho-\sigma\Vert_1 \le \sqrt{1-F^2(\rho,\sigma)}.
\end{equation}

Let $\ha$ be an unbounded positive semi-definite operator and $\rho$ a state. Let $\indic_{[0,n]}(\ha)$ be the spectral projection of $\ha$ onto energies of at most $n$, we then define
\[ \tr(\ha\rho) \coloneqq  \sup_{n \in \mathbb N} \tr\left(\ha \indic_{[0,n]}(\ha) \rho\right) \in [0,\infty].\]
A single mode of a continuous variable quantum system can be described by bosonic annihilation and creation operators, $\hat{a}$ and $\hat{a}^{\dagger}$, respectively. They satisfy the canonical commutation relation $\left[ \hat{a}, \hat{a}^{\dagger} \right] = 1$. The so-called bosonic number operator is then defined as $\hat{N} = \hat{a}^{\dagger} \hat{a}$. In this paper, we focus on infinite-dimensional quantum systems governed by the Hamiltonian $\ha = \hat{N}$.

The entropies considered in this paper are defined in Section \ref{subsec:entropies} below. In addition to them, the binary entropy is defined as $h(\varepsilon) \coloneqq -\varepsilon \log \varepsilon - (1-\varepsilon) \log (1-\varepsilon).$ We also use the functions $f_1(x)=-x\log(x)$ and $f_{\alpha}(x)=x^{\alpha}$ for $\alpha \in (0,1).$

We denote by $C^{\alpha}(I)$ the space of $\alpha$-H\"older continuous functions on $I$. We denote by $\Lambda_{\omega}$ the space of functions continuous with respect to the modulus of continuity $\omega.$ Definitions can be found in Section \ref{sec:QCM}. The integrated modulus of continuity $\omega^*$ is defined in Section \ref{subsec:2.3}.

We write $f=\mathcal O(g)$ to indicate that there is $C>0$ such that $\vert f(x)\vert \le C \vert g(x) \vert$ for all $x$ in the common domain of functions $f$ and $g$. We write $f=o(g)$ as $x$ to $a$ if  $\lim_{x \to a} \left|\frac{f(x)}{g(x)}\right| = 0.$ The error function is denoted by $\operatorname{erf}(z)=\frac{2}{\sqrt{\pi}}\int_0^z e^{-t^2} \ dt.$ The cardinality of a set $A$ is denoted by $\vert A \vert.$

Finally, we recall the concept of {\em{asymptotic tightness:}} As explained, for example in~\cite{S21b}, a continuity bound $\sup_{x,y \in S_a} |f(x) - f(y)|\le C_a(x,y)$ depending on a parameter $a$ (with $S_a$ being a set), is called asymptotically tight for large $a$ if 
$$ \limsup_{a \to \infty} \sup_{x,y \in S_a} \frac{|f(x) - f(y)|}{C_a(x,y)}=1.$$

\subsection{Entropies}
\label{subsec:entropies}
Let $(\Omega,\Sigma, \mathbb P)$ be a probability space. Let $X,Y: \Omega\rightarrow Z$ be discrete random variables, with $Z=\bigcup_{i \in \mathbb N_0} \{z_i\}$ and probabilities $p_X(i) \coloneqq \mathbb P(X=z_i)$ for $i \in \mathbb N_0.$ 
\begin{defi}[Entropies: Discrete random variables]
\label{defi:disc_ent}
The Shannon entropy of $X$ is then defined as 
\[H(X) = -\sum_{i=0}^{\infty} p_X(i) \log(p_X(i)),\]
and for $\alpha\in (0,1) \cup (1,\infty)$, we introduce the $\alpha$-Tsallis entropy
\[T_{\alpha}(X) =\frac{\sum_{i=0}^{\infty} p_X(i)^{\alpha}-1}{1-\alpha}, \]
and $\alpha$-R\'enyi entropy
\[R_{\alpha}(X) =\frac{\log\Big(\sum_{i=0}^{\infty} p_X(i)^{\alpha}\Big)}{1-\alpha}. \]
\end{defi}
Let $p_{XY}(i,j) \coloneqq \mathbb P(X=z_i,Y=z_j)$ be the joint distribution, then the conditional entropy of $X$ given $Y$ is defined as
\[ H(X\vert Y)_p = -\sum_{(i,j) \in \mathbb N_0^2} p_{XY}(i,j) \log(p_{XY}(i,j)/p_Y(i)).\]

\begin{defi}[Quantum entropies]
\label{defi:q_ent}
The von Neumann entropy of a state $\rho$ of an infinite-dimensional quantum system with Hilbert space $\mathcal{H}$ is defined as 
\[S(\rho) = -\tr(\rho\log(\rho)),\]
and for $\alpha\in (0,1)$, we introduce the quantum $\alpha$-Tsallis entropy
\[T_{\alpha}(\rho) =\frac{\tr(\rho^{\alpha})-1}{1-\alpha}, \]
and the quantum $\alpha$-R\'enyi entropy
\[R_{\alpha}(\rho) =\frac{\log\tr\Big(\rho^{\alpha}\Big)}{1-\alpha}. \]
\end{defi}

\subsection{Quantitative continuity measures for functions} 
\label{sec:QCM} A general quantitative measure of continuity for functions is the so-called \emph{modulus of continuity} $\omega:[0,\infty] \to [0,\infty]$. 
We say that a function $f:(X,\vert \cdot \vert_X) \to (Z,\vert \cdot \vert_Z)$, where $X,Z$ are subsets of a normed space, admits $\omega$ as a modulus of continuity if $\vert f\vert_{\Lambda_{\omega}} \coloneqq \sup_{x\neq y} \frac{\vert f(x)-f(y)\vert_Z}{\omega(\vert x-y \vert_X)}$ is finite, in which case 
\[ \vert f(x)-f(y)\vert_Z \le \vert f\vert_{\Lambda_{\omega}} \omega (\vert x-y\vert_X) \text{ for all }x,y \in X.\]
The function $\omega$ is assumed to be monotonically-increasing, positive-definite, and subadditive. The vector space of such functions $f$ for which $\vert f \vert_{\Lambda_{\omega}}$ is finite is denoted by $\Lambda_{\omega}.$
The space of functions $C^{\alpha} \coloneqq \Lambda_{\omega_{\alpha}}$, with $\alpha \in (0,1)$ and modulus of continuity $\omega_{\alpha}(t) =t^{\alpha}$ is called the space of \emph{$\alpha$-H\"older continuous functions}.
We also introduce the space of functions $\Lambda_{\operatorname{AL}}$ that are characterized by a modulus of continuity $$\omega_{\operatorname{AL}}(t) = \begin{cases} -t\log(t) &\text{ for }t \in [0,1/e] \text{ and } \\
1/e &\text{ otherwise. }\end{cases}.$$ 
This is the class of local \emph{almost Lipschitz functions}. Here, we employ a cut-off at $t=1/e$, which is the $t$ at which $-t\log(t)$ attains its maximum, as the maximum distance between discrete probability distributions in total variation distance and states in trace norm is always bounded by two.
Many functions related to entropies fall in some of these two spaces $C^{\alpha}$ or $\Lambda_{\operatorname{AL}}$. 

The function $f_1(x)= -x\log(x)$, related to the Shannon entropy is almost Lipschitz and the functions associated with $\alpha$-R\'enyi entropies $f_{\alpha}(x)=x^{\alpha} \in C^{\gamma}([0,\tau])$ for $\gamma \le \alpha<1, \tau>0$ are H\"older continuous. 
Indeed, that $\omega_{\operatorname{AL}}$ is a modulus of continuity for $f_1$ follows from \cite[Theo. 17.3.3]{CT06}
\begin{equation}
\label{eq:almostlip}
\vert f_1 (x) - f_1 (y) \vert \le f_1(\vert x-y\vert) \text{ for all }x,y \in [0,1/e].
\end{equation}

For functions $f_{\alpha}$ we find that since $\vert x^{\alpha}-y^{\alpha}\vert \le \vert x-y\vert^{\alpha}$ then $\vert f_{\alpha}\vert_{\Lambda_{\omega_{\alpha}}} = 1.$
Since $f_1$ is smooth on $[1/e,\infty)$ this clearly implies that $f_1$ is almost Lipschitz and therefore in particular H\"older continuous. That almost Lipschitz functions are always H\"older continuous follows from the limit $\lim_{t \to 0}\omega_{\text{AL}}(t)/t^{
\alpha}=0.$

\subsection{Quantitative continuity measures for functions of operators} 
\label{subsec:2.3}
The analysis of continuity estimates for functions of self-adjoint operators is more subtle as the regularity of functions is usually not preserved at the operator level. By this we mean that for instance Lipschitz functions $f$ are in general not operator Lipschitz, \textit{i.e.} functions $f:\mathbb R \rightarrow \mathbb R$ for which there is $C>0$ such that 
\[ \vert f(x)-f(y) \vert \le C \vert x-y\vert \text{ for all }x,y \in \mathbb R \] do not satisfy $\Vert f(A)-f(B) \Vert \le C' \Vert A-B\Vert$ for bounded self-adjoint $A,B$ for any $C'>0.$ An example of a function that is Lipschitz but not operator Lipschitz is the function $f(x)=\vert x \vert.$

However, the study of continuity estimates, for functions in $\Lambda_{\omega}$, can be transformed to estimate functions of bounded self-adjoint operators. This has been established in a series of papers starting from \cite{AP,AP1}. The figure of merit for estimates on the operator level is then the \emph{integrated modulus of continuity} $\omega^*(t)=t \int_t^{\infty} \frac{\omega(x)}{x^2} \ dx,$ with $t \ge 0$ which for the cases we considered before reads
\begin{equation}
\omega^*_{\alpha}(t) = \frac{t^{\alpha}}{1-\alpha}
\end{equation}
and
\begin{equation}
\begin{split}
\omega^*_{\operatorname{AL}}(t) =\begin{cases} &\Big(e^2-\frac{1}{2}\Big)t  + \frac{\log(t)^2t}{2}, \text{ for } t \in [0,1/e] \text{ and }\\
&e, \text{ otherwise}.\end{cases}
\end{split}
\end{equation}
Details on the specifics of the continuity bounds obtained in \cite{AP,AP1} are provided in the beginning of Section \ref{sec:QEntropies}.

\subsection{Analytical background}
\begin{prop}[Courant-Fisher]
\label{prop:CFT}
 Let $\hat H$ be a self-adjoint operator that is bounded from below with purely discrete spectrum.
Let $E_1\le E_2\le E_3\le\cdots$ be the eigenvalues of $\hat H$; then 
$$E_n=\min_{\psi_1,\ldots,\psi_{n}}\max\{\langle\psi,\hat H\psi\rangle:\psi\in\operatorname{lin}\{\psi_1,\ldots,\psi_{n}\}, \, \| \psi \| = 1\},$$
where $\operatorname{lin}$ denotes the linear hull. 
In particular, this implies that for $\pi_1,..,\pi_n$ being the first $n$ eigen-projections of $\hat H$, counting multiplicity, and $\tilde \pi_1,..,\tilde \pi_n$ any other mutually orthogonal one-dimensional projections, then we have for any $\lambda_1\ge...\ge\lambda_n \ge 0$ that 
\[ \tr\Big(\hat H\sum_{i=1}^n \lambda_i \pi_i\Big) \le \tr\Big(\hat H\sum_{i=1}^n \lambda_i \tilde{\pi_i}\Big).\]
\end{prop}

We also recall the following simple fact:
\begin{lemm}
\label{lemm:auxlemm}
Let $f:(0,1) \rightarrow [0,\infty)$ be a measurable function such that $\int_0^1 \frac{f(t)}{t} dt < \infty$. Then there is a sequence of $t_n \downarrow 0$ such that $f(t_n) \downarrow 0.$ 
\end{lemm}
\begin{proof}
If such a sequence does not exist, then $f(t) \ge \epsilon>0$ for $t\in (0,\delta)$ and $\delta>0$. Thus, $\int_0^1 \frac{f(t)}{t} \ dt \ge \int_0^{\delta} \frac{\epsilon}{t} \ dt = \infty$
which contradicts our assumption.
\end{proof}

\section{Main results}

\subsection{Shannon entropies}
\label{sec:Shannon}
\subsubsection{A mean-constrained Fano's inequality}

Let $(\Omega, \Sigma,\mathbb P)$ be a probability space. Let $X,Y : \Omega \rightarrow \mathbb N_0$ be a pair of random variables.
Let $\mathcal{P}$ denote the set of joint probability distributions $p_{XY}$ on $\mathbb N_0^2$. 
Since the alphabet is infinite, we need to impose a constraint on the random variables in order for the entropies of the distributions in $\mathcal{P}$ to be finite. 
We choose a constraint on the means of the marginals, \textit{i.e}, for $p_{XY} \in \mathcal{P}$,
\begin{equation} \label{eq:defMomentConstraint}
	\mathbb{E}(X) = \sum_{n=0}^{\infty} n p_X(n) \leq E, \quad 	\mathbb{E}(Y) = \sum_{m=0}^{\infty} m p_Y(m) \leq E, 
\end{equation}
for some finite $E > 0$. 

\begin{theo}[Mean-constrained Fano's inequality on $\mathbb N_0$] \label{theo:Fanob}
Let $X$ and $Y$ be a pair of random variables taking values in $\mathbb N_0$, with joint probability distribution $p_{XY}$, satisfying the 
constraint ${\mathbb{E}}(X) \le E$ for some $0<E< \infty$.
Then for all $\varepsilon \in [0,E/(E+1)]$ the following inequality holds:
\begin{equation}\label{F-theob}
\pr{X}{Y}{p} \leq h(\varepsilon) + E h({\varepsilon}/E),
\end{equation}
where $\varepsilon  \coloneqq  \mathbb{P}(X \ne Y)$. Furthermore, for any given $E$, the inequality is tight for all $\varepsilon \in [0,E/(E+1)]$.
\end{theo}
\begin{proof}
Let us define the set $Z \coloneqq \{ m \in \mathbb N_0 : p_{XY}(m,m) \geq p_{XY}(0,m) \}$ and the new random variable $X'$ such that the joint probability distribution $p_{X'Y}$ is defined as follows: for all $m \in \mathbb N_0$,
\begin{equation*}
	p_{X'Y}(0,m)  \coloneqq  \begin{cases}
		p_{XY}(m,m), & \forall \, m \in Z, \\
		p_{XY}(0,m), & \mathrm{else},
	\end{cases}
\end{equation*}
for all $m \in \mathbb N$,
\begin{equation*}
	p_{X'Y}(m,m)  \coloneqq  \begin{cases}
		p_{XY}(0,m), & \forall \, m \in Z, \\
		p_{XY}(m,m), & \mathrm{else},
	\end{cases}
\end{equation*}
and $p_{X'Y}(n,m)  \coloneqq  p_{XY}(n,m)$, otherwise.
First note that $\mathbb E(X')\le \mathbb E(X) \le E$ since
\eqal{ \label{upxpb}
    \mathbb{E}(X') & = \sum_{n \in \mathbb N_0} n p_{X'}(n)  = \sum_{n \in \mathbb N_0} n \sum_{m \in \mathbb N_0} p_{X'Y}(n,m) \\
    & = \sum_{n \in Z} n \, p_{X'Y}(n,n) + \sum_{n \in \mathbb N_0 \setminus Z} n \, p_{X'Y}(n,n) \\
    & \quad \quad + \sum_{n \in \mathbb N_0} n \sum_{\substack{m \in \mathbb N_0 \\ m \ne n}} p_{X'Y}(n,m) \\
    & = \sum_{n \in Z} n \, p_{XY}(0,n) + \sum_{n \in \mathbb N_0 \setminus Z} n \, p_{XY}(n,n) \\
    & \quad \quad + \sum_{n \in \mathbb N_0} n \sum_{\substack{m \in \mathbb N_0 \\ m \ne n}} p_{XY}(n,m) \\
    & \leq \sum_{n \in Z} n \, p_{XY}(n,n) + \sum_{n \in \mathbb N_0 \setminus Z} n \, p_{XY}(n,n) \\
    & \quad \quad + \sum_{n \in \mathbb N_0} n \sum_{\substack{m \in \mathbb N_0 \\ m \ne n}} p_{XY}(n,m)\\
    &=\mathbb{E}(X)  \leq E,
}
where the first inequality follows from the definition of the set $Z$, while the second follows from the constraint on the mean: ${\mathbb{E}}(X) \le E$.
Secondly, note that $H(X'|Y) = H(X|Y)$, since $H(X'Y) = H(XY)$, which can be easily checked. 
 
Using the fact that conditioning decreases entropy, we end up with
\begin{equation} \label{eqfb}
     H(X|Y) = H(X'|Y) \leq H(X').
\end{equation}
Hence, to complete the proof, it suffices to find an upper bound on $H(X')$.
Define $\varepsilon' \coloneqq 1 - p_{X'}(0)$ and note that
\eqal{
    \varepsilon' & = 1 - \sum_{m \in \mathbb N_0} p_{X'Y}(0,m) \\
    & = 1 - \sum_{m \in Z} p_{X'Y}(0,m)- \sum_{m \in \mathbb N_0 \setminus Z} p_{X'Y}(0,m) \\
    & = 1 - \sum_{m \in Z} p_{XY}(m,m) - \sum_{m \in \mathbb N_0 \setminus Z} p_{XY}(0,m).
}
Since $p_{XY}(m,m) \geq p_{XY}(0,m)$ if and only if $m \in Z$, we have that $p_{XY}(m,m) < p_{XY}(0,m)$ for all $m \in \mathbb N_0 \setminus Z$, so that
\eqal{
\varepsilon' & \le 1 - \sum_{m \in Z} p_{XY}(m,m) - \sum_{m \in \mathbb N_0 \setminus Z} p_{XY}(m,m) \\
& = 1 - \sum_{m \in \mathbb N_0} p_{XY}(m,m) = \mathbb{P}(X\ne Y) = \varepsilon,
}
which in turn implies that $\sum_{n=1}^{\infty} p_{X'}(n) = \varepsilon' \le \varepsilon$.
Hence,
\eqal{ \label{eqlb}
	H(X') & = -(1-\varepsilon') \log(1-\varepsilon') - \sum_{n=1}^{\infty} p_{X'}(n) \log p_{X'}(n) \\
	& = h(\varepsilon') + \varepsilon' \log \varepsilon' -  \sum_{n=1}^{\infty} p_{X'}(n) \log p_{X'}(n) \\
	& = h(\varepsilon') - \sum_{n=1}^{\infty} p_{X'}(n) \log \frac{p_{X'}(n)}{\varepsilon'} \\
	& = h(\varepsilon') - \varepsilon' \sum_{n=0}^{\infty} r(n) \log r(n)= h(\varepsilon') + \varepsilon' H(R)
}
where $R$ is a random variable taking values in $\mathbb N_0$ with distribution
\[
r(n) \equiv 	\mathbb P(R=n)  \coloneqq  \frac{p_{X'}(n+1)}{\varepsilon'}, \quad \forall n \in \mathbb N_0.\]
	To find an upper bound on $H(R)$, we estimate ${\mathbb{E}}(R)$ using \eqref{upxpb}
\eqal{ \label{upb}
  \mathbb{E}(R)
  &= \frac{1}{\varepsilon'}\sum_{n = 1}^{\infty} (n-1) p_{X'}(n) = \frac{1}{\varepsilon'} \mathbb{E}(X') - 1  \leq \frac{E}{\varepsilon'} - 1.
}
 It is known that the geometric distribution achieves maximum entropy among all distributions of a given mean on $\mathbb N_0$. This allows us to upper bound $H(R)$ with the entropy of an appropriate geometric random variable. Let $Z$ denote a geometric random variable with parameter $p \in (0,1)$, that is $\mathbb P(Z=k)=(1-p)^k p$ for $k \in \mathbb N_0$. Its mean and its Shannon entropy are respectively given by 
\begin{align}
{\mathbb{E}}(Z)= \frac{1-p}{p}, \quad {\hbox{and}} \quad H(Z) = \frac{h(p)}{p},
\end{align}
and the entropy is a decreasing function of the parameter $p$.
By setting ${\mathbb{E}}(Z) = E/\varepsilon' - 1$ (which is the upper bound on ${\mathbb{E}}(R)$) we obtain
\begin{equation}\label{upfb}
    p = \frac{1}{{\mathbb{E}}(Z) + 1}= \frac{\varepsilon'}{E}
    \text{ and hence }
    H(R) \leq H(Z) = \frac{E}{\varepsilon'} h(\varepsilon'/E).
\end{equation}

From \eqref{eqfb}, \eqref{eqlb} and \eqref{upfb}, it follows that
\begin{equation}\label{F-theob2}
\pr{X}{Y}{p} \leq h(\varepsilon') + E h({\varepsilon'}/E),
\end{equation}
with $\varepsilon' \leq \varepsilon$. 
By studying its derivative, it is easy to see that the right hand side of \eqref{F-theob2} is an increasing function of $\varepsilon'$ for all $\varepsilon' \in [0, E/(E+1)]$. As a result, we end up with
\begin{equation}
    \pr{X}{Y}{p} \leq h(\varepsilon) + E h({\varepsilon}/E),
\end{equation}
for all $\varepsilon \in [0, E/(E+1)]$, which proves \eqref{F-theob}.

In order to see that the above inequality is tight for $\varepsilon \in [0, E/(E+1)]$, consider the random variables $(X^*,Y^*)$ characterized by the joint probability distribution $p_{X^*Y^*}$ which is defined as follows:
\begin{equation} \label{eq:probaTightb}
    p_{X^*Y^*}(n,m) \coloneqq \begin{cases}
    1 - \varepsilonilon & \quad \mathrm{if} \, n = 0, m = 0, \\
    \varepsilonilon \, w(n-1) & \quad \forall \, n \in \mathbb N, \, \mathrm{for} \, m = 0, \\
    0, & \quad \mathrm{else},
    \end{cases}
\end{equation}
where $\{w(n)\}_{n \in \mathbb N_0}$ is the probability distribution of the geometric random variable $W$ with mean ${\mathbb{E}}(W) = E/\varepsilon - 1$. First note that
\eqal{ \label{eq:probaTightb2}
    \mathbb{E}(X^*) & = \sum_{n \in \mathbb N_0} n \sum_{m \in \mathbb N_0} p_{X^*Y^*}(n,m)  = \sum_{n \in \mathbb N_0} n p_{X^*Y^*}(n,0) \\
    & = \varepsilon \sum_{n \in \mathbb N_0} (n+1) w(n)  = \varepsilon \left( \mathbb{E}(W) + 1 \right)  = E.
}
Secondly, note that
\[    \mathbb{P}(X^* \ne Y^*)  = 1 - \mathbb{P}(X^* = Y^*)  = 1 - \sum_{n \in \mathbb N_0} p_{X^*Y^*}(n,n)  = \varepsilon.\]
Finally, since $H(Y^*) =0$,
\eqal{ \label{eq:probaTightb3}
    H(X^*|Y^*)_p & = H(X^*Y^*)  \\
    & = - (1-\varepsilon) \log (1-\varepsilon) - \sum_{n \in \mathbb N_0} \varepsilon w(n) \log \left( \varepsilon w(n) \right) \\
    & = h(\varepsilon) - \varepsilon \sum_{n \in \mathbb N_0} w(n) \log w(n) \\
    & = h(\varepsilon) + \varepsilon H(W)  = h(\varepsilon) + E h(\varepsilon/E).
}
This proves the theorem.
\end{proof}

Consider a function $f : \cA \rightarrow \mathbb{R}_+$ that satisfies the following properties:
\begin{equation} \label{eq:Boltzmann}
    \begin{aligned}
        (i) & \, f(0) = 0, \\
        (ii) & \, f \textrm{ is a non-decreasing function}, \\
        (iii) & \, \textrm{for any } \lambda<0, \quad e^{-\lambda_0} \coloneqq \sum_{x \in \cA} e^{\lambda f(x)} < \infty,
    \end{aligned}
\end{equation}
with $\lambda_0 \in \mathbb{R}$, so that $\{ e^{\lambda_0 + \lambda f(x)} \}_{x \in \cA}$ is a valid probability distribution with finite Shannon entropy. Then for every $E > f(0)$ we consider the set of all random variables with probability distributions $\{ p(x) \}_{x \in \cA}$ satisfying the constraint $\sum_{x \in \cA}  f(x) p(x) \le E$. It is known that the maximal Shannon entropy is achieved (uniquely) by the random variable $W_E$ taking values in $\cA$ with probability distribution $\{w_{E}(x)\}_{x \in \cA}$ given by
\begin{equation} \label{eq:defW1}
    w_E(x) \coloneqq e^{\lambda_0(E) + \lambda(E) f(x)}, \quad \forall x \in \cA.
\end{equation}
Here, $\lambda_0(E)$ and $\lambda(E)$ are chosen such that
\begin{equation} \label{eq:condGibbsClassical2}
    \sum_{x \in \cA} w_{E}(x) = 1, \qquad \con{W_E} = E.
\end{equation}
We then refer to $W_E$ as the random variable that achieves the maximum Shannon entropy for $\con{W_E} = E$. Similarly, define the function $\tilde{f} : \cA \rightarrow \mathbb{R}_+$ as
\begin{equation} \label{eq:defTildef}
    \tilde{f}(x) = f(x+1).
\end{equation}
In that case, $\tilde{f}$ also satisfies properties $(ii)$ and $(iii)$ of~\eqref{eq:Boltzmann} with $f$ being replaced by $\tilde{f}$. This then guarantees the existence of the random variable $\tilde{W}_E$ that achieves the maximum Shannon entropy for $\cont{\tilde{W}_E} = E$. With these considerations, Theorem \ref{theo:Fanob} can be generalized by considering a constraint of the form $\con{X} \leq E$ instead of a mean constraint.
\begin{theo}[Fano's inequality for countably infinite alphabet with a general constraint] \label{theo:FanoG}
Let $X$ and $Y$ be a pair of random variables taking values in $\cA$, having a joint distribution $p_{XY}$, and satisfying the constraint $\con{X} \le E$ for some $E \in [f(1)\varepsilon,\infty)$, where $\eps := \mathbb{P}(X \ne Y)$, and the function $f$ satisfies the properties of~\eqref{eq:Boltzmann}.
Then the following inequality holds:
\begin{equation}\label{F-theoG}
\pr{X}{Y}{p} \leq h(\eps) + \eps H(\tilde{W}_{E/\eps}),
\end{equation}
for $\eps \in \mathcal{E}(E)$, where $\mathcal{E}(E) \subseteq [0,1]$ contains the values of $\eps$ for which the right-hand side of \eqref{F-theoG} is a non-decreasing function of $\eps$,
and $\tilde{W}_{E/\eps}$ is the random variable that achieves the maximum Shannon entropy for $\cont{\tilde{W}_{E/\eps}} = E/\eps$, with $\tilde{f}$ being defined in \eqref{eq:defTildef}.
Furthermore, the inequality is tight for $\eps \in \mathcal{E}(E)$.
\end{theo}
Since the proof of Theorem \ref{theo:FanoG} closely follows the lines of that of Theorem \ref{theo:Fanob}, we omit its proof.
Note that the right-hand side of \eqref{F-theoG} can be shown to be a concave function of $\eps$ that reaches $0$ when $\eps=0$ and is non-negative for all $\eps \in [0,1]$. Consequently, there always exists some $\eps^* \in (0,1]$ such that $\mathcal{E}(E) = [0,\eps^*]$.

\subsubsection{A mean-constrained continuity bound for the entropy of random variables with a countably infinite alphabet}

In \cite{Sason2013}, Sason exploited the concept of maximal coupling of random variables in order to rederive the standard continuity bound for the Shannon entropy (\textit{i.e.},~for random variables with a finite alphabet) from Fano's inequality. In this section, we extend his proof to obtain a continuity bound for the Shannon entropy of random variables with a countably infinite alphabet, under mean constraints. In order to do so, we make use of the mean-constrained  Fano inequality given in Theorem~\ref{theo:Fanob}, as well as the concept of maximum coupling.

Recall that a \emph{coupling} of a pair of random variables $(X,Y)$ is a pair of random variables $(\hat{X},\hat{Y})$ with the same marginal probability distributions as of $(X,Y)$. 

\begin{defi}[Maximal coupling]For a pair of random variables $(X,Y)$, a coupling $(\hat{X},\hat{Y})$ is called a \emph{maximal coupling} if $\mathbb{P}(\hat{X}=\hat{Y})$ attains its maximal value among all the couplings of $(X,Y)$. 
\end{defi}

\begin{theo}[Mean-constrained Shannon entropy continuity bound] \label{theo:Continuityb}
Let $X$ and $Y$ be a pair of random variables taking values in $\mathbb N_0$, with respective probability distributions $p_X$ and $p_Y$, and satisfying the constraints ${\mathbb{E}}(X) \le E$ and  ${\mathbb{E}}(Y) \le E$ for some $0<E <\infty$.
Then for all $\eps \in [0,E/(E+1)]$ the following inequality holds:
\begin{equation} \label{eq:theo:Continuityb}
    \left| \hr{X}{p} - \hr{Y}{p} \right| \leq h(\eps) + E h\left( \eps/E \right),
\end{equation}
where $\eps \coloneqq \Vert X-Y\Vert_{\operatorname{TV}}$. Furthermore, for any given $0<E<\infty$, the inequality is tight for all $\eps \in [0,E/(E+1)]$.
\end{theo}
 \begin{proof}
	Again, we take inspiration from the proof of Theorem~3 in \cite{Sason2013}. Let $(\hat{X},\hat{Y})$ be a maximal coupling of $(X,Y)$, and $p_{XY}$ be the corresponding joint probability distribution. Since $\hr{X}{p} = \hr{\hat{X}}{p}$ and $\hr{Y}{p} = \hr{\hat{Y}}{p}$, we have, for all $\eps \in [0,E/(E+1)]$,
\begin{equation*}
    \begin{split}
		\left| \hr{X}{p} - \hr{Y}{p} \right| & = \left| \hr{\hat{X}}{p} - \hr{\hat{Y}}{p} \right| \\
		& = \left| \pr{\hat{X}}{\hat{Y}}{p} - \pr{\hat{Y}}{\hat{X}}{p} \right| \\
		& \leq \max \left\lbrace \pr{\hat{X}}{\hat{Y}}{p}, \pr{\hat{Y}}{\hat{X}}{p} \right\rbrace \\
		& \stackrel{(1)}{\leq} h(\varepsilon') + E h\left( \varepsilon'/E \right) \\
		& \stackrel{(2)}{=} h(\eps) + E h\left( \eps/E \right)
    \end{split}
\end{equation*}
where $\varepsilon'  \coloneqq  \mathbb{P}(\hat{X} \ne \hat{Y})$, and $(1)$ follows from Theorem~\ref{theo:Fanob}, while $(2)$ is a consequence of the fact that if $(\hat{X},\hat{Y})$ is a maximal coupling of $(X,Y)$ then
\[		\mathbb{P}(\hat{X} \neq \hat{Y}) = \Vert X-Y\Vert_{\operatorname{TV}}.\]
A proof of this for random variables with finite state space has been given in \cite{Sason2013} and the proof for infinite state alphabets is a straightforward adaptation of the argument. 

In order to see that the above inequality is tight for $\eps \in [0, E/(E+1)]$, consider the random variables $X^*$ and $Y^*$ characterized by the probability distributions $p_{X^*}$ which is defined as follows:
\begin{equation} \label{eq:probaTightb4}
    p_{X^*}(n) \coloneqq \begin{cases}
    1 - \eps & \quad \mathrm{if} \, n = 0, \\
    \eps \, w(n-1) & \quad \mathrm{else},
    \end{cases}
\end{equation}
where $\{w(n)\}_{n \in \mathbb N_0}$ is the probability distribution of the geometric random variable $W$ characterized by a mean-value ${\mathbb{E}}(W) = E/\eps - 1$, and $p_{Y^*}$ which is defined as follows:
\[  p_{Y^*}(m) \coloneqq \begin{cases}
    1 & \quad \mathrm{if} \, m = 0, \\
    0 & \quad \mathrm{else}.
    \end{cases}\]
Note that $p_{X^*}$ and $p_{Y^*}$ correspond to the two marginals of $p_{X^*Y^*}$ defined in \eqref{eq:probaTightb}. From this and \eqref{eq:probaTightb2}, we know that $\mathbb{E}(X^*) = E$. From \eqref{eq:probaTightb3}, we have that $H(X^*) = h(\eps) + E h(\eps/E)$. Obviously, we also have that $\mathbb{E}(Y^*) = 0 < E$ and $H(Y^*) = 0$. Finally, it is easy to see that $\Vert X^*-Y^*\Vert_{\operatorname{TV}} = \eps$. This proves the theorem.
\end{proof}

If a function $f : \cA \rightarrow \mathbb{R}_+$ satisfies the properties of~\eqref{eq:Boltzmann}, one has the following generalization of Theorem \ref{theo:Continuityb} for a constraint of the form $\con{X} \leq E$. Since its proof is analogous to that of Theorem~\ref{theo:Continuityb}, we omit the proof.
\begin{theo}[Continuity bound for the Shannon entropy of random variables under a general constraint] \label{theo:ContinuityG}
Let $X$ and $Y$ be a pair of random variables taking values in $\cA$, with respective probability distributions $p_X$ and $p_Y$, and satisfying the constraints $\con{X} \le E$ and $\con{Y} \le E$ for some $E \in [f(1)\varepsilon,\infty)$, where  $\eps := \Vert X-Y\Vert_{\operatorname{TV}}$, and the function $f$ satisfies the properties of~\eqref{eq:Boltzmann}.
Then the following inequality holds:
\begin{equation} \label{eq:theo:ContinuityG}
    \left| \hr{X}{p} - \hr{Y}{p} \right| \leq h(\eps) + \eps H(\tilde{W}_{E/\eps}),
\end{equation}
for $\eps \in \mathcal{E}(E)$, where $\mathcal{E}(E) \subseteq [0,1]$ contains the values of $\eps$ for which the right-hand side of \eqref{eq:theo:ContinuityG} is a non-decreasing function of $\eps$,
and $\tilde{W}_{E/\eps}$ is the random variable that achieves the maximum Shannon entropy for $\cont{\tilde{W}_{E/\eps}} = E/\eps$, with $\tilde{f}$ being defined in \eqref{eq:defTildef}.
\end{theo}

\subsection{von Neumann entropies} \label{sec:ContinuityQuantum}

We now lift the classical entropy continuity estimate in Theorem~\ref{theo:Continuityb} to general density operators, i.e. positive trace-class operators on a separable, infinite-dimensional Hilbert space with unit trace. The assumption that the state space of the classical random variables is $\mathbb N_0$ enforces us now to take as the Hamiltonian the single-mode number operator, $\hat{N}=a^\dagger a$.
\begin{theo}[von Neumann entropy continuity bound] \label{theo:ContinuityQuantum}
Let $\ha$ denote the number operator, and let $\rho$ and $\sigma$ be two quantum states on a separable, infinite-dimensional Hilbert space $\cH$, satisfying the energy constraints $\tr(\ha\rho), \tr(\ha\sigma) \leq E$, for some $0<E<\infty$, such that
\begin{equation}
    \frac{1}{2} ||\rho-\sigma||_1 \leq \td.
\end{equation}
Then for all $\eps \in [0, E/(1+E)]$ the following inequality holds:
\begin{equation} \label{eq:theo:ContinuityQuantum}
	|S(\rho) - S(\sigma) |\leq h(\td) + E h\left( \td/E \right).
\end{equation}
Furthermore, for any given $0<E<\infty$, the inequality is tight for all $\eps \in [0,E/(E+1)]$.
\end{theo}
Before providing the proof of the above theorem, let us recall the known uniform energy-constrained continuity bound for the von Neumann entropy of infinite-dimensional quantum states obtained in \cite{Winter2016} by Winter. In his paper, Winter considers a Hamiltonian $\ha$ which has a discrete spectrum, is bounded from below\footnote{In fact, he fixes the ground state energy to be zero.} and satisfies the following so-called Gibbs Hypothesis.
\begin{defi}\cite{Winter2016}\label{defi:Gibbs}.  A Hamiltonian $\ha$ is said to satisfy the Gibbs Hypothesis if for any $\beta>0$, $Z(\beta) \coloneqq \tr(e^{-\beta \ha}) < \infty$, so that $e^{-\beta \ha}/Z(\beta)$ is a valid quantum state with finite entropy.
\end{defi}
\begin{figure*}
\begin{center}
\includegraphics[width=7.5cm,height=4.5cm]{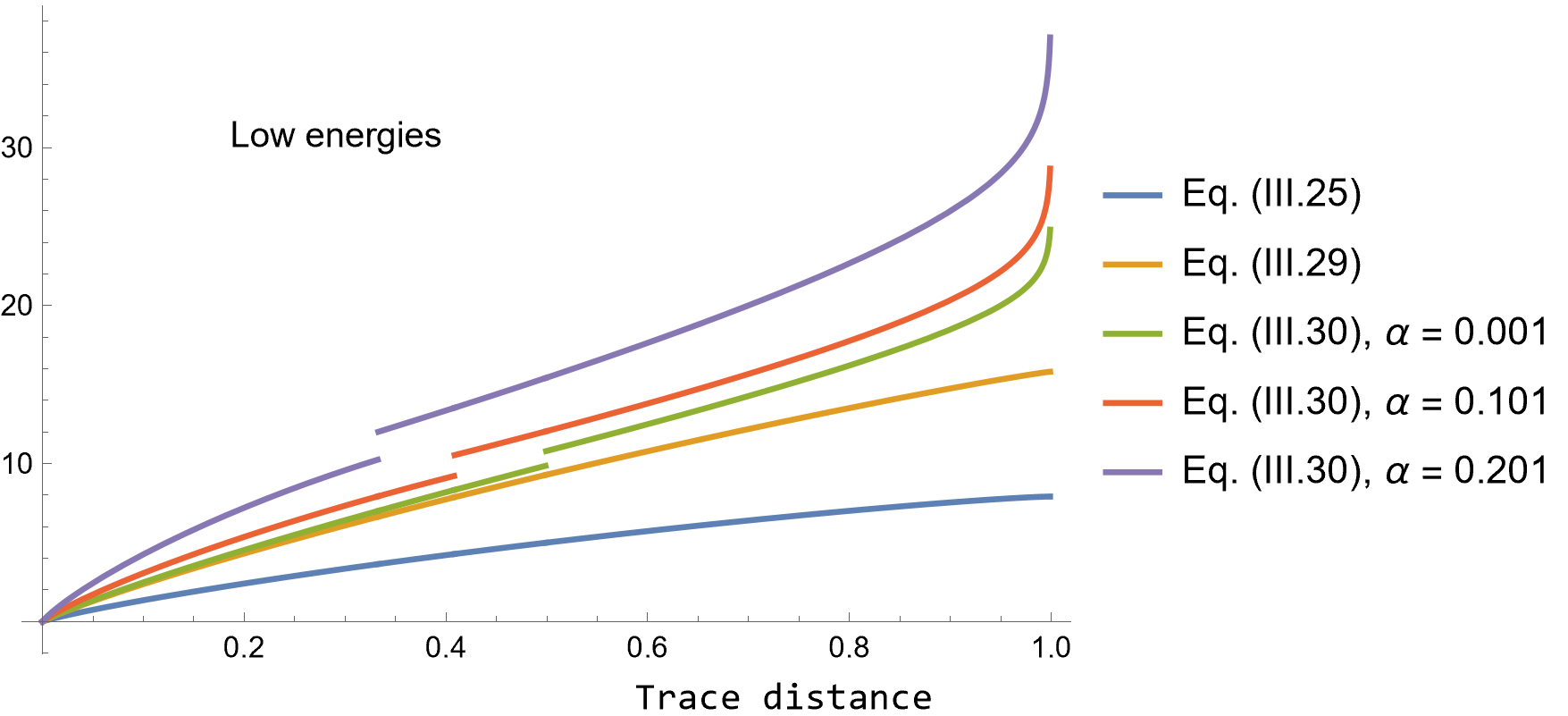} \hspace{2cm} \includegraphics[width=7.5cm,height=4.5cm]{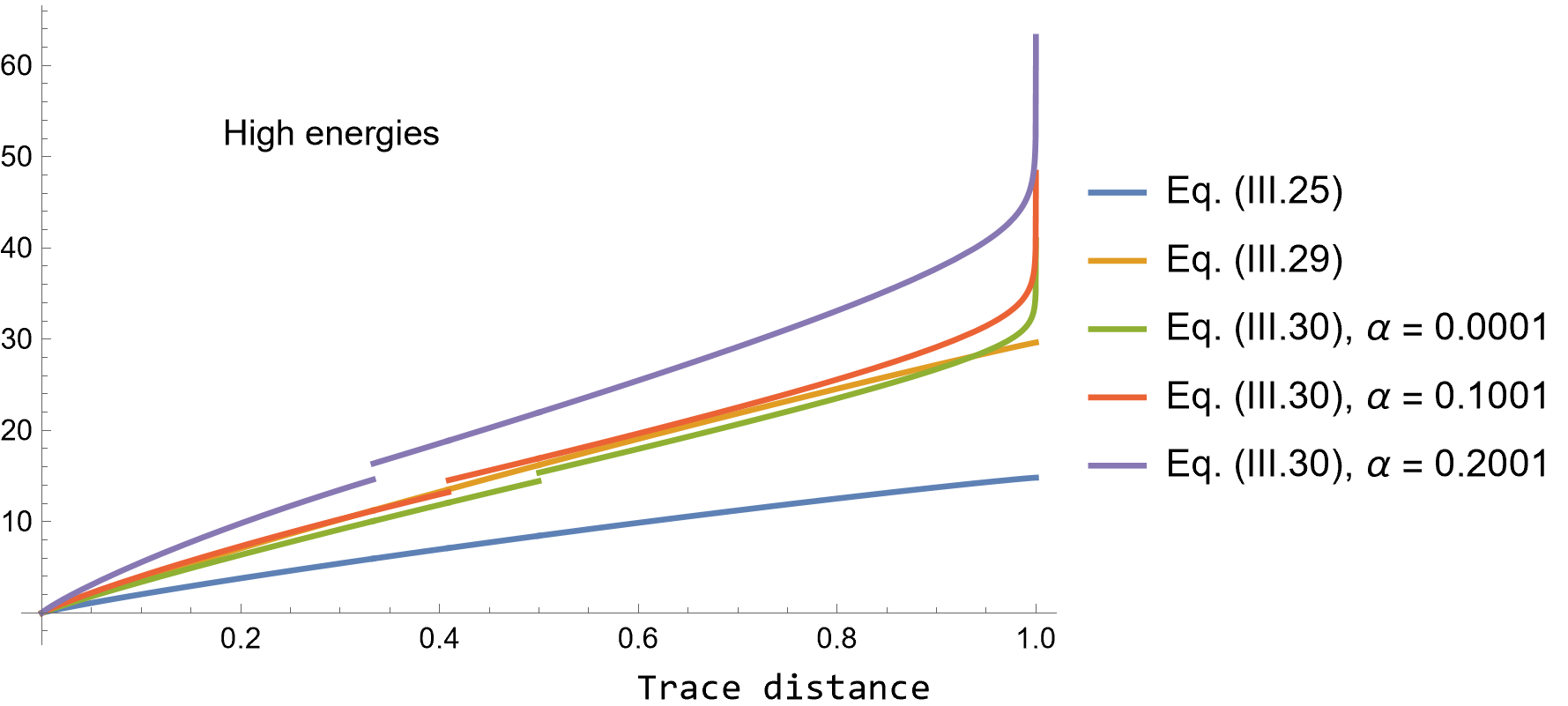}
\end{center}
\caption{We compare our tight bound \eqref{eq:theo:ContinuityQuantum} to the bound~\eqref{eq:BoundWinter3} by Winter for general Hamitonians specialized to the single-mode number operator and also to~\eqref{eq:BoundWinter2} for different values of $\alpha.$ Due to the piecewise definition of~\eqref{eq:BoundWinter2} the latter curves also show a jump discontinuity. We see that for a wide range of $\alpha$ and low energies, \eqref{eq:BoundWinter3} outperforms \eqref{eq:BoundWinter2}, but for high energies, there exist values of $\alpha$ for which \eqref{eq:BoundWinter2} outperforms~\eqref{eq:BoundWinter3}.\label{fig:fig2}}
\end{figure*}
If $\ha$ satisfies the Gibbs hypothesis, then for every $E > \inf \Spec(\ha)$, among all states satisfying the energy constraint $\tr(\ha \rho) \le E$, the maximal entropy is achieved (uniquely) by the Gibbs state:
\begin{equation}
\label{eq:Gibbshyp}
    \gamma(E) = \frac{ e^{- \beta(E) \ha}}{Z(\beta(E))},
\end{equation}
where the parameter $\beta(E)$ is decreasing with $E$ and is determined by the equality
\begin{equation} \label{eq:WinterBeta}
    \tr(e^{- \beta \ha}(\ha-E)) = 0.
\end{equation}
For such a Hamiltonian, Winter proved the following energy-constrained continuity bound~\cite{Winter2016}: for any two states $\rho$ and $\sigma$ on a separable, infinite-dimensional Hilbert space with $\tr(\rho \hat H), \tr(\sigma \hat H) \leq E$ and $\frac{1}{2} || \rho-\sigma||_1 \leq \td \leq 1$,
\begin{equation} \label{eq:BoundWinter}
    |S(\rho)-S(\sigma)| \leq h(\td) + 2 \td S \left( \gamma(E/\td) \right).
\end{equation}
In the case in which the Hamiltonian is the number operator $\hat{N}$ corresponding to a single mode, the bound \eqref{eq:BoundWinter} reduces to
\begin{equation} \label{eq:BoundWinter3}
    |S(\rho)-S(\sigma)| \leq h(\eps) + 2 \left(E + \eps \right) \, h\left( \frac{\eps}{E+\eps} \right).
\end{equation}
Winter also gets an additional bound for the single mode number operator $\hat{N}$ (see Lemma~18 of~\cite{Winter2016}): 
\eqal{ \label{eq:BoundWinter2}
    & |S(\rho)-S(\sigma)| \\
    & \quad \leq \eps \left(\frac{1+\alpha}{1-\alpha} + 2\alpha \right) \left[ \log(E+1) + \log \frac{\eps}{\alpha(1-\eps)}\right] \\
   & \quad \quad + 3 \left( \frac{1+\alpha}{1-\alpha} + 2\alpha\right) \tilde{h}\left(\frac{1+\alpha}{1-\alpha} \eps \right)=:K(\varepsilon,\alpha,E),
}
where $\alpha \in (0, 1/2)$ and $\tilde{h}(x) = h(x) $ for $x \le 1/2$ and $\tilde{h}(x) = 1 $ for $x \ge 1/2$.
As we explain below and as has been already observed in \cite{W15}, the bound \eqref{eq:BoundWinter2} is asymptotically tight for a suitable joint limit of $\alpha \to 0$ and $E \to \infty$ for every $\varepsilon>0$ fixed. In addition, it is unknown so far how to optimize the choice of $\alpha$ in \eqref{eq:BoundWinter2} and for many choices of $\alpha$, the bound \eqref{eq:BoundWinter} seems to be better than the bound~\eqref{eq:BoundWinter2}; see Fig. \ref{fig:fig1}.

\begin{figure*}
\begin{center}
\includegraphics[width=7.5cm,height=5.5cm]{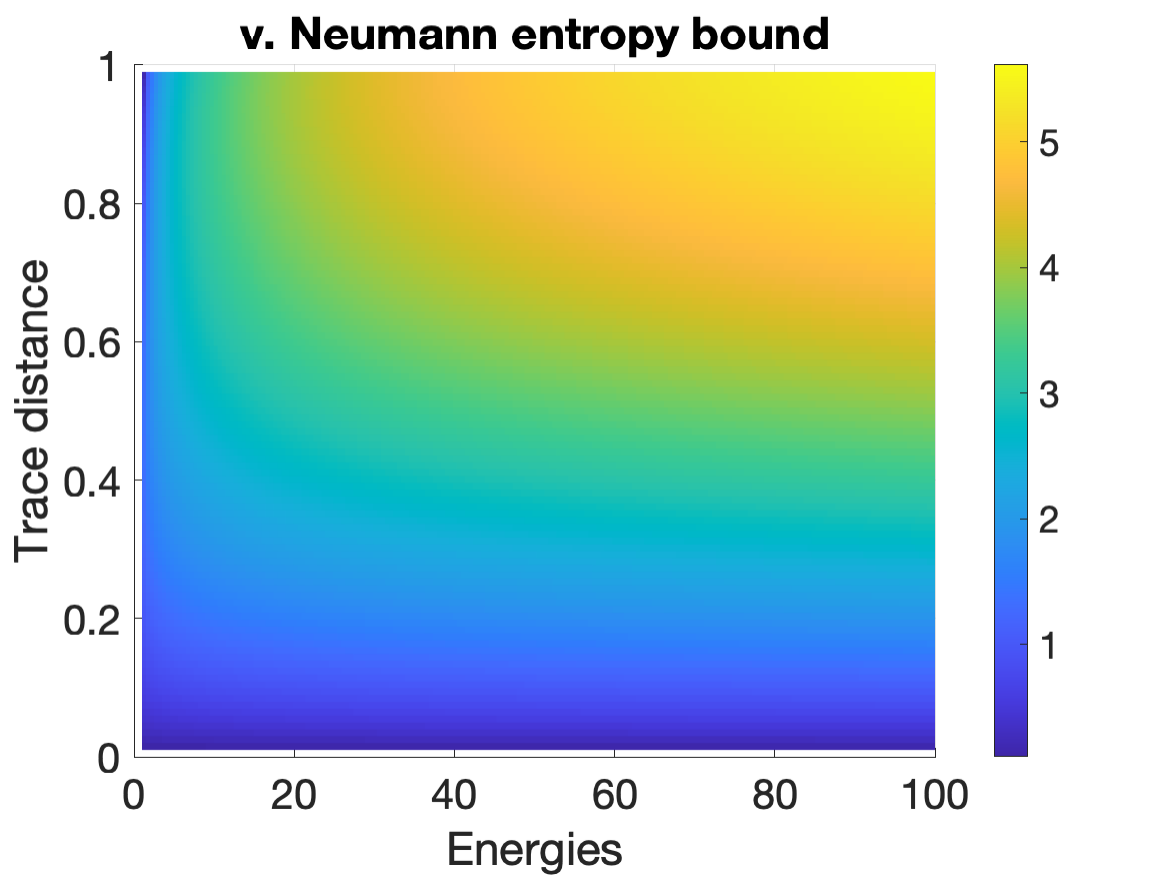}\quad \includegraphics[width=7.5cm,height=5.5cm]{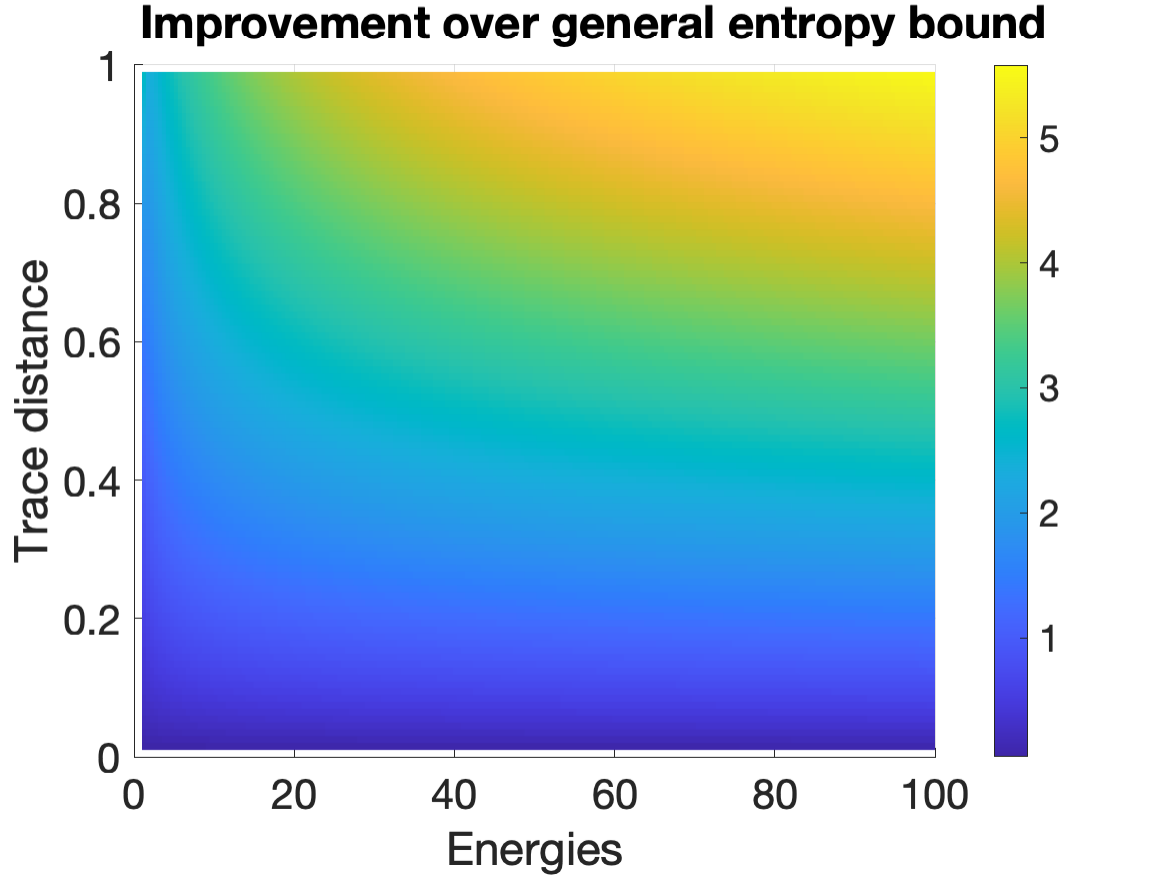} \\[15pt] \includegraphics[width=7.5cm,height=5.5cm]{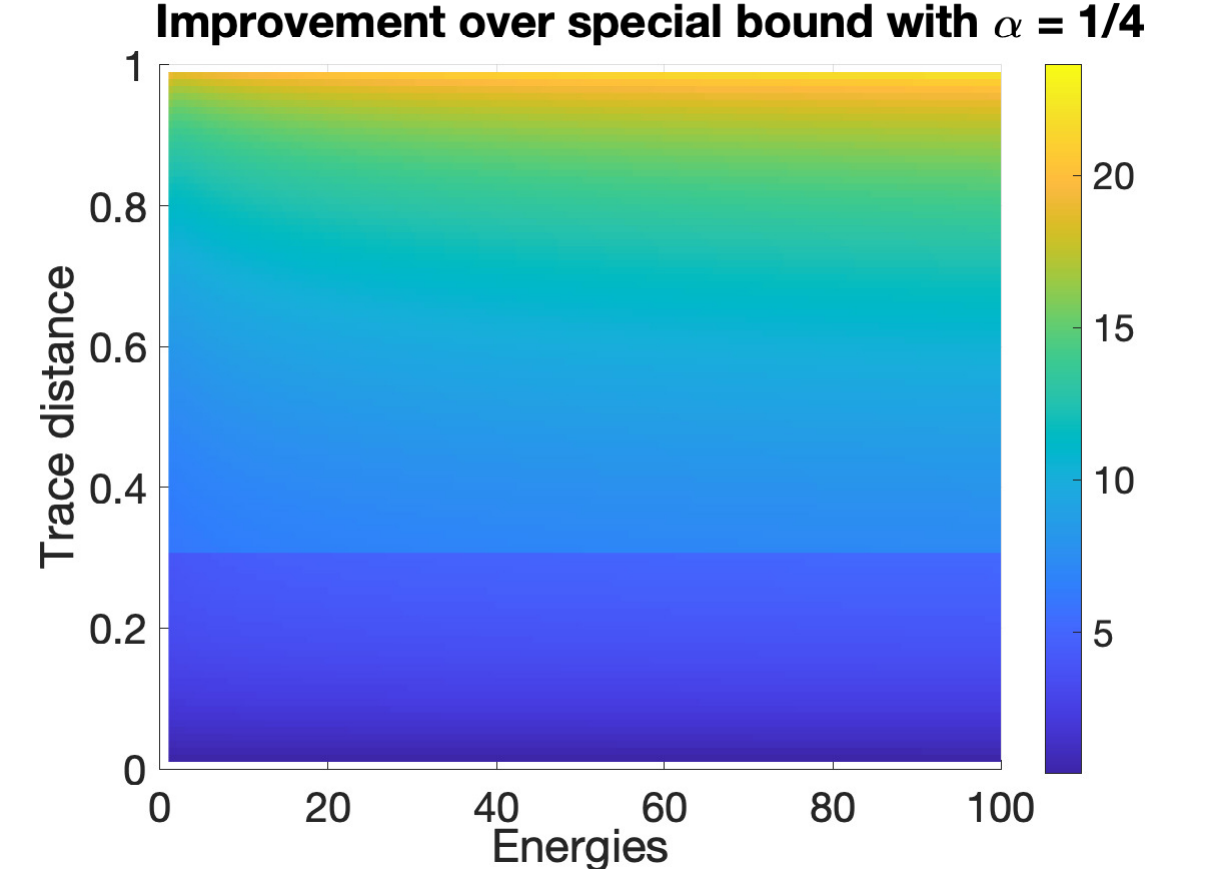}\quad \includegraphics[width=7.5cm,height=5.5cm]{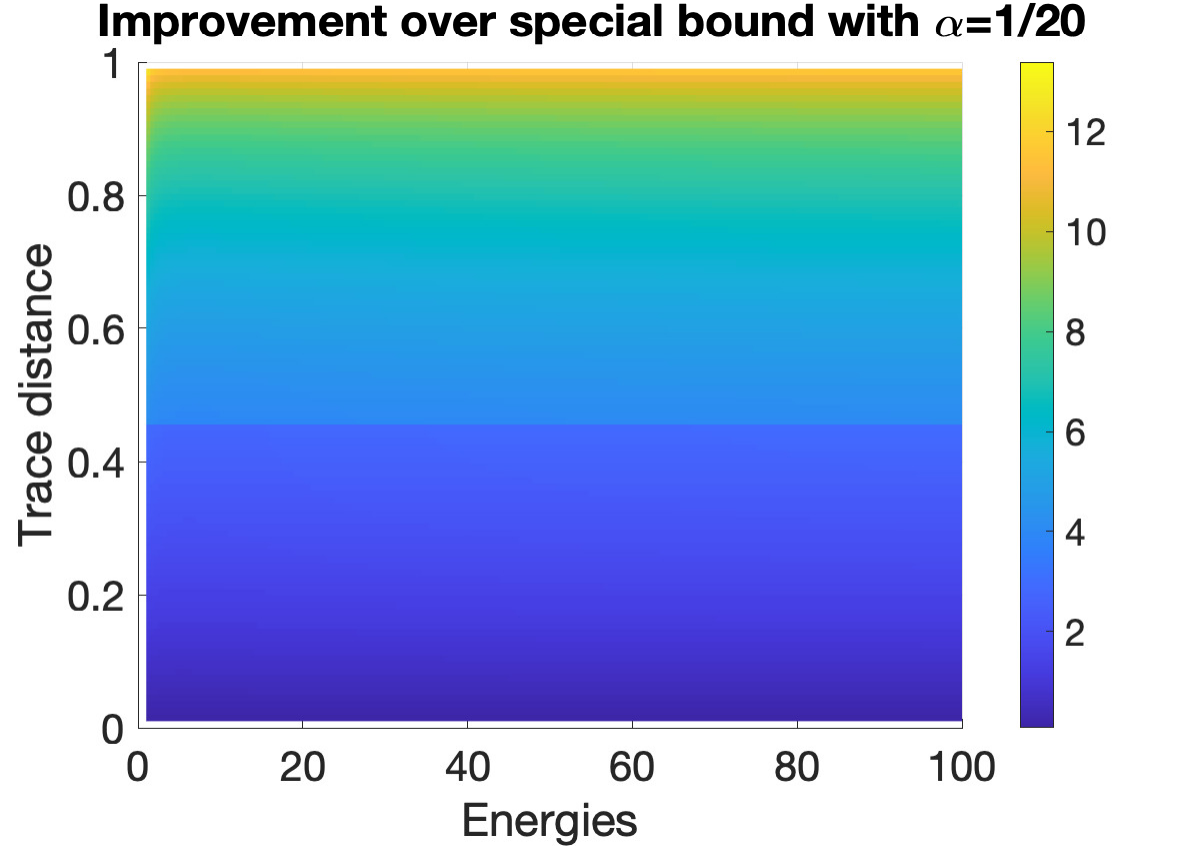}
\end{center}
\caption{The upper left figure illustrates the right-hand side of our tight bound \eqref{eq:theo:ContinuityQuantum}. The upper right figure illustrates the difference of the bound \eqref{eq:BoundWinter} obtained by Winter \cite{W15} to our bound on the von Neumann entropy.
The lower two figures compare the bound \eqref{eq:BoundWinter2} found by Winter, for fixed $\alpha$, to our bound. The improvement is in all cases particularly  significant for high energies and large trace distances.\label{fig:fig1}}
\end{figure*}

In contrast, our bound given in Theorem~\ref{theo:ContinuityQuantum} is tight for all values of the trace distance $\eps \in [0, E/(E+1)]$) for any given $0<E<\infty$.
\begin{rem}[Asymptotic tightness]
We now address the issue of asymptotic tightness: the difference of entropies for states $\rho,\sigma$ with eigenvalues according to the probability distributions optimizing the estimate in the proof of Theorem \ref{theo:Continuityb} is given by \eqref{eq:probaTightb3} 

\[ \vert S(\rho)-S(\sigma)\vert = h(\varepsilon) + E h(\varepsilon/E).\]

We then study the behaviour of \eqref{eq:BoundWinter2} for a one-parameter family $E(n)=e^{n}, \alpha(n) =n^{-1}$ for any fixed $\varepsilon>0$. Thus, using Taylor expansion
\begin{equation}
    \begin{split}
        & E(n) h(\varepsilon/E(n)) \\
        & = -\varepsilon \log(\varepsilon e^{-n}) - e^n(1-\varepsilon e^{-n})\log(1-\varepsilon e^{-n}) \\
        & = -\varepsilon\log(\varepsilon)+n\varepsilon - (e^n-\varepsilon)\log(1-\varepsilon e^{-n}) \\
        & = n\varepsilon + o(n) \text{ as } n \to \infty.
    \end{split}
\end{equation} 
In particular, since $h(\varepsilon)$ is just a constant, we have 
\begin{equation}
\label{eq:entdiff} \vert S(\rho)-S(\sigma)\vert = n\varepsilon + o(n)\text{ as } n \to \infty.
\end{equation}
Now observe that 
\[\eps \left(\frac{1+\frac{1}{n}}{1-\frac{1}{n}} + \frac{2}{n}\right) \log(e^n+1) = \eps n + o(n)\text{ as } n \to \infty.\]
Since this is the leading order term in the right-hand side of \eqref{eq:BoundWinter2}, we thus also have that  
\begin{equation}
\label{eq:Kcap}
K(\varepsilon,\alpha(n),E(n)) =\varepsilon n+o(n)\text{ as } n \to \infty.
\end{equation}
Hence, by combining \eqref{eq:entdiff} with \eqref{eq:Kcap} 
\[\lim_{n \to \infty} \frac{\vert S(\rho)-S(\sigma)\vert}{K(\varepsilon,\alpha(n),E(n))} = 1\]
which shows asymptotic tightness.
\end{rem}
We now provide a proof of Theorem~\ref{theo:ContinuityQuantum}.
\begin{proof}
Let the spectral decompositions of $\rho$ and $\sigma$ be given by
\begin{equation}
    \rho = \sum_{n=0}^{\infty} r(n) \proj{\phi_n}, \qquad \sigma = \sum_{n=0}^{\infty} s(n) \proj{\psi_n}.
\end{equation}
Consider the passive states
\begin{equation}
    \rho^{\downarrow} = \sum_{n=0}^{\infty} r^{\downarrow}(n) \proj{n}, \qquad \sigma^{\downarrow} = \sum_{n=0}^{\infty} s^{\downarrow}(n) \proj{n},
\end{equation}
where the states $\ket{n}$ for $n = 0, 1, \cdots$ are eigenstates of the Hamiltonian $\ha$, $\{r^{\downarrow}(n)\}_{n \in \mathbb N_0}$ represents the distribution containing the non-zero elements of $\{r(n)\}_{n \in \mathbb N_0}$ arranged in non-increasing order, \textit{i.e.}, $r^{\downarrow}(n) \geq r^{\downarrow}(n+1)$ for all $n \in \mathbb N_0$, and similarly for $\{s^{\downarrow}(n)\}_{n \in \mathbb N_0}$.
We obviously have that $S(\rho^{\downarrow}) = S(\rho)$ and $S(\sigma^{\downarrow}) = S(\sigma)$, so that
\begin{equation}
	|S(\rho) - S(\sigma)| = |S(\rho^{\downarrow}) - S(\sigma^{\downarrow})|.
\end{equation}
Furthermore, from the Courant-Fischer theorem in Proposition~\ref{prop:CFT}, we have $\tr(\ha \rho^{\downarrow}) \leq \tr(\ha \rho) \leq E$, $\tr(\ha \sigma^{\downarrow}) \leq \tr(\ha \sigma) \leq E$ \footnote{This can also be proved using Ky Fan's Maximum Principle~\cite[Lemma IV.9]{palma}} and
\begin{equation}
    \eps' \coloneqq \frac{1}{2} ||\rho^{\downarrow}-\sigma^{\downarrow}||_1 \leq \frac{1}{2} ||\rho-\sigma||_1 \leq \td.
\end{equation}
The above inequality also follows from Mirsky's inequality~\cite{Mirsky}, see also \cite[(1.22)]{simon2005trace} for a version in infinite dimensions. Let $X$ and $Y$ denote random variables on $\mathbb N_0$, with probability distributions $r \equiv \{r^{\downarrow}(n)\}_{n \in \mathbb N_0}$ and $s \equiv \{s^{\downarrow}(n)\}_{n \in \mathbb N_0}$, respectively. Then ${\mathbb{E}}(X) = \Tr(\ha \rho^{\downarrow}) \leq E$, ${\mathbb{E}}(Y) = \Tr(\ha \sigma^{\downarrow}) \leq E$, $H(X) = S(\rho^{\downarrow})$, $H(Y) = S(\sigma^{\downarrow})$ and $\Vert X-Y\Vert_{\operatorname{TV}} = \eps' \leq \eps$.
Using Theorem~\ref{theo:Continuityb}, we have
\eqal{
	|S(\rho) - S(\sigma)| & = |H(X) - H(Y)|  \leq h\left( \eps' \right) + E h\left( \eps'/E \right).
}
As mentioned before, it is easy to see by analyzing its derivative that the right-hand side of the last inequality in the above equation is an increasing function of $\eps'$ for all $\eps' \in [0, E/(E+1)]$. As a result, we end up with
\begin{equation}
    |S(\sigma) - S(\rho)| \leq h\left( \td \right) + E h\left( \td/E \right),
\end{equation}
for all $\eps \in [0, E/(E+1)]$.

In order to see that the above inequality is tight for $\eps \in [0, E/(E+1)]$, consider the quantum states $\rho^* \coloneqq \sum_{n=0}^{\infty} p_{X^*}(n) \proj{n}$ and $\sigma^* = \proj{0}$ where $p_{X^*}$ is the probability distribution defined in \eqref{eq:probaTightb4}. From this, we have that $S(\rho) = h\left( \td \right) + E h\left( \td/E \right)$ and $\tr(\ha\rho) \leq E$. Obviously, we also have that $S(\sigma) = 0$ and $\tr(\ha\sigma) = 0 < E$. Finally, it is trivial to see that $\frac{1}{2} ||\rho^*-\sigma^*||_1 = \td$. This proves the theorem.
\end{proof}

Consider a Hamiltonian $\ha$ on a separable infinite-dimensional Hilbert space $\cH$, with ground state energy $0$, that satisfies the Gibbs Hypothesis. As a consequence, it has a discrete spectrum of finite multiplicity and can be represented as \cite{Shirokov2016}
\begin{equation} \label{eq:specDecompH}
    \ha = \sum_{n=0}^{\infty} f(n) \proj{e_n},
\end{equation}
where $\{ \ket{e_n} \}_{n \in \mathbb{N}_0}$ is an orthonormal basis and the eigenvalues $f(n)$ for $n \in \mathbb{N}_0$ are such that $f(n)\leq f(n+1)$. Define the function $\tilde{f} : \cA \rightarrow \mathbb{R}_+$ as $\tilde{f}(x) = f(x+1)$ and the Hamiltonian $\hb$ on $\cH$ as
\begin{equation}
    \hb = \sum_{n=0}^{\infty} \tilde{f}(n) \proj{e_n}.
\end{equation}
The following generalization of Theorem \ref{theo:ContinuityQuantum} can be shown using a similar proof.
\begin{theo}[Continuity bound for the von Neumann entropy of energy-constrained states for general Hamiltonians] \label{theo:ContinuityQuantumG}
Let the Hamiltonian $\ha$ on $\cH$, with ground state energy $0$, satisfy the Gibbs Hypothesis. Let $\rho$ and $\sigma$ be two quantum states on $\cH$, satisfying the energy constraints $\tr(\ha\rho), \tr(\ha\sigma) \leq E$, for some $E \in [f(1)\varepsilon,\infty)$, such that
\begin{equation}
    \frac{1}{2} ||\rho-\sigma||_1 \leq \td,
\end{equation}
with $\td \in [0,1]$.
Then the following inequality holds:
\begin{equation} \label{eq:theo:ContinuityQuantumG}
	|S(\rho) - S(\sigma) |\leq h(\eps) + \eps S(\tilde{\gamma}(E/\eps)),
\end{equation}
for $\eps \in \mathcal{E}(E)$, where $\mathcal{E}(E) \subseteq [0,1]$ contains the values of $\eps$ for which the right-hand side of \eqref{eq:theo:ContinuityQuantumG} is a non-decreasing function of $\eps$, and $\tilde{\gamma}(E/\eps)$ denotes the Gibbs state of energy $E/\eps$ corresponding to the Hamiltonian $\hb$.
Furthermore, the inequality is tight for $\eps \in \mathcal{E}(E)$.
\end{theo}
Note that the right-hand side of \eqref{eq:theo:ContinuityQuantumG} is the same function as the right-hand side of \eqref{F-theoG}. Consequently, there always exists some $\eps^* \in (0,1]$ such that $\mathcal{E}(E) = [0,\eps^*]$ in the above theorem.
In order to better understand the size and scaling of the interval $\mathcal{E}(E),$ it is at least necessary to understand the behaviour of the entropy of the Gibbs state $\tilde{\gamma}(E/\eps)$ for small values of $\eps.$ This question has been addressed in \cite[Theorem $3$]{SBND}. This result states that for any Hamiltonian satisfying a certain spectral condition that is virtually met for all common Hamiltonians satisfying the Gibbs hypothesis, there exists a parameter $\eta$ such that 

\begin{equation}
\label{eq:bound_5} 
S(   \tilde{\gamma}(E)) = \eta \log(E)(1+o(1)) \text{ as } E \to \infty.
\end{equation}
Since for $\eps \in [0,1/2]$, $\eps \mapsto h(\eps)$ is monotone, the right hand side of \eqref{eq:theo:ContinuityQuantumG} is increasing, at least for $\eps \in [0,1/2]$ if $[0,1/2] \ni \eps \mapsto \varepsilon S(\tilde{\gamma}(E/\eps))$ increases.

Using the asymptotic result \eqref{eq:bound_5}, we see that for either large energies or $\eta$ sufficiently small, the monotonicity of the right-hand side of \eqref{eq:theo:ContinuityQuantumG} is expected to hold on a large interval $\mathcal E(E)$ of admissible $\eps.$ For Schr\"odinger operators with bounded potential on bounded domains $\Omega \subset \mathbb R^n$, it was shown in \cite{SBND} that $\eta = n/2$ for example. Therefore, one might expect. in addition, that the interval does not increase for Hamiltonians on multi-dimensional domains.

In principle, one could hope to turn \eqref{eq:bound_5} into a quantitative estimate that provides quantitative estimates on the size of $\mathcal E(E).$ In practice, it seems more effective to directly analyze concrete Hamiltonians and exploit specific features about their spectra, such as their high-energy limits, to understand the size of $\mathcal E(E).$


\subsection{Classical $\alpha$-R\'enyi and $\alpha$-Tsallis entropies}
\label{sec:classicalent}
Let $(\Omega, \Sigma,\mathbb P)$ be a probability space. We start by considering two random variables $X,Y: \Omega \rightarrow Z$ where we shall assume that $Z$ is either a discrete countably infinite set or a measurable subset $D\subset \RR^d.$ In the latter case, we assume that $X$ and $Y$ possess probability densities $\mu_X,\mu_Y,$ respectively.

\subsubsection{Discrete random variables} 
Let $Z \coloneqq \bigcup_{i \in \mathbb N} \{ z_i \}$ and $p_X(z_i) \coloneqq \mathbb P(X=z_i).$ We are interested in studying for $\alpha$-H\"older continuous functions $f:[0,1] \rightarrow \RR$ the quantities with $f(p_X) \coloneqq (f(p_X(z_1)),f(p_X(z_2)),f(p_X(z_3)),...) \in \ell^1(\mathbb N)$ with
\begin{equation}\label{Tf}
(Tf)(X) \coloneqq  \sum_{i \in \mathbb N} f(p_X(z_i)),\end{equation}
where $H(X)=Tf_1(X)$ is the \emph{Shannon entropy} of the random variable $X$, $R_{\alpha}(X)=\frac{\log(Tf_{\alpha}(X))}{1-\alpha}$ its $\alpha$-\emph{R\'enyi entropy}, and $T_{\alpha}(X)=\frac{Tf_{\alpha}(X)-1}{1-\alpha}$ its $\alpha$-\emph{Tsallis entropy}.

\begin{prop}
\label{prop:simple}
Let $f \in C^{\alpha}([0,1])$,  $\alpha \in (0,1)$, and $X,Y$ random variables with discrete countably infinite state space $Z$. Let $w=(w_i)$ be a sequence of positive weights such that $\left(w_i^{-\frac{\beta}{1-\alpha}}\right) \in \ell^1(\mathbb N)$ for some $\beta$ such that $0 < \beta < \alpha$ and $(w_i p_X(z_i)),(w_i p_Y(z_i)) \in \ell^1(\mathbb N)$. Then we have the following continuity bound: 
\eqal{ \label{eq32}
& \Vert f(p_X)-f(p_Y) \Vert_{\ell^1} \\
& \le 2^{\alpha} \vert f \vert_{\Lambda_{\omega_{\alpha}}} \Vert X-Y\Vert^{\beta}_{\operatorname{TV}(w)} \Vert X-Y\Vert_{\operatorname{TV}}^{\alpha-\beta} \Vert (w_i^{-\frac{\beta}{1-\alpha}})\Vert^{1-\alpha}_{\ell^1}, 
}
where by the triangle inequality $\vert Tf(X)-Tf(Y) \vert \le \Vert f(p_X)-f(p_Y) \Vert_{\ell^1}.$
\end{prop}
\noindent Note that the conditions $\left(w_i^{-\frac{\beta}{1-\alpha}}\right) \in \ell^1(\mathbb N)$ for some $\beta < \alpha$ and $(w_i p_X(z_i)),(w_i p_Y(z_i)) \in \ell^1(\mathbb N)$ replace the moment constraint in Eq.~\ref{eq:theo:Continuityb}.
\begin{proof}
Using H\"older continuity of $f$, we find

\[ \Vert f(p_X)-f(p_Y) \Vert_{\ell^1} \le \vert f \vert_{\Lambda_{\omega_{\alpha}}}  \sum_{i \in \mathbb N}  w_i^{-\beta} w_i^{\beta} \vert p_X(z_i)-p_Y(z_i) \vert^{\alpha} .\]
Choosing $p=1/\alpha$ and its H\"older conjugate $q=1/(1-\alpha)$, we have by H\"older's inequality
\eqal{
    & \Vert f(p_X)-f(p_Y) \Vert_{\ell^1} \\
    & \le \vert f \vert_{\Lambda_{\omega_{\alpha}}} \left( \sum_{i \in \mathbb N} w_i^{\frac{\beta}{\alpha}} \vert p_X(z_i)-p_Y(z_i) \vert \right)^{\alpha} \left( \sum_{i \in {\mathbb N}} w_i^{-\frac{\beta}{1-\alpha}}\right)^{1-\alpha}.
    \label{329}
}
Applying H\"older's inequality to the second summand on the right hand side of the above line with the choice $p=\alpha/\beta$ and $q = \alpha/(\alpha - \beta)$, we obtain
\eqal{
    \sum_{i \in {\mathbb N}} w_i^{\beta/\alpha} \vert p_X(z_i)-p_Y(z_i) \vert & \le \left(\sum_{i \in {\mathbb N}} w_i \vert p_X(z_i)-p_Y(z_i) \vert\right)^{\frac{\beta}{\alpha}} \\
    & \times \left(\sum_{i \in {\mathbb N}} \vert p_X(z_i)-p_Y(z_i) \vert\right)^{\frac{\alpha - \beta}{\alpha}},
}
which together with (\ref{329}) yields

\eqal{
& \Vert f(p_X)-f(p_Y) \Vert_{\ell^1} \\
& \le \vert f \vert_{\Lambda_{\omega_{\alpha}}}  \Vert (w_i^{-\frac{\beta}{1-\alpha}})\Vert^{1-\alpha}_{\ell^1} \times \left\lVert \vert (p_X(z_i)-p_Y(z_i)) \right\rVert_{\ell^1(w_i)}^{\beta} \\
& \quad \quad \times \left\lVert \vert (p_X(z_i)-p_Y(z_i)) \right\rVert_{\ell^1}^{\alpha-\beta}.
}
Expressing the $\ell^1$ distances in terms of the total variation distances using (\ref{TV}) then yields the claim.
\end{proof}
\begin{rem}
Let $Z \subset \mathbb Z$ and $w_i \coloneqq \max\{\vert z_i\vert,1\}$ then the condition $(w_i p_X(z_i)),(w_i p_Y(z_i)) \in \ell^1(\mathbb N),$ is equivalent to the existence of a first moment. If $Z=\mathbb N$, then $ \Vert (w_i^{-\frac{\beta}{1-\alpha}})\Vert^{1-\alpha}_{\ell^1} = \zeta(\tfrac{\beta}{1-\alpha})^{1-\alpha}$ where $\zeta$ is the Riemann zeta function and we may choose $\beta\in (1-\alpha,\alpha).$ In particular, the constraints restrict us to choosing $\alpha \in (1/2,1).$
\end{rem}
As an immediate corollary from Proposition \ref{prop:simple}, we then find:
\begin{corr}
\label{corr:contdisc}
Let $\alpha \in (0,1)$ and $X,Y$ be random variables with discrete countably infinite state space $Z$. Let $(w_i)$ be a sequence of positive weights such that $\left( w_i^{-\frac{\beta}{1-\alpha}} \right) \in \ell^1(\mathbb N)$ for some $\beta < \alpha$ and $(w_i p_X(z_i)),(w_i p_Y(z_i)) \in \ell^1(\mathbb N)$. Then we have the following continuity bounds: 
The $\alpha$-Tsallis entropy satisfies 
\eqal{
& \vert T_{\alpha}(X)-T_{\alpha}(Y) \vert \\
& \le \frac{2^{\alpha}   }{1-\alpha}\Vert X-Y\Vert^{\beta}_{\operatorname{TV}((w_i))} \Vert X-Y\Vert_{\operatorname{TV}}^{\alpha-\beta} \Vert (w_i^{-\frac{\beta}{1-\alpha}})\Vert^{1-\alpha}_{\ell^1}.
}
The $\alpha$-R\'enyi entropy satisfies 
\eqal{
& \vert R_{\alpha}(X)-R_{\alpha}(Y) \vert \\
& \le \frac{2^{\alpha} } {1-\alpha}\Vert X-Y\Vert^{\beta}_{\operatorname{TV}((w_i))} \Vert X-Y\Vert_{\operatorname{TV}}^{\alpha-\beta} \Vert (w_i^{-\frac{\beta}{1-\alpha}})\Vert^{1-\alpha}_{\ell^1}.
}
\end{corr}
\begin{proof}
Recall that $(Tf_{\alpha})(X) \ge 1,$ which allows us to use that $\vert \log(x)-\log(y) \vert \le \vert x-y\vert$ for $x,y \ge 1,$ for the R\'enyi entropy. The result then immediately follows from the triangle inequality and Proposition \ref{prop:simple}.
\end{proof}

\begin{figure*}
\begin{center}
\includegraphics[width=7.5cm]{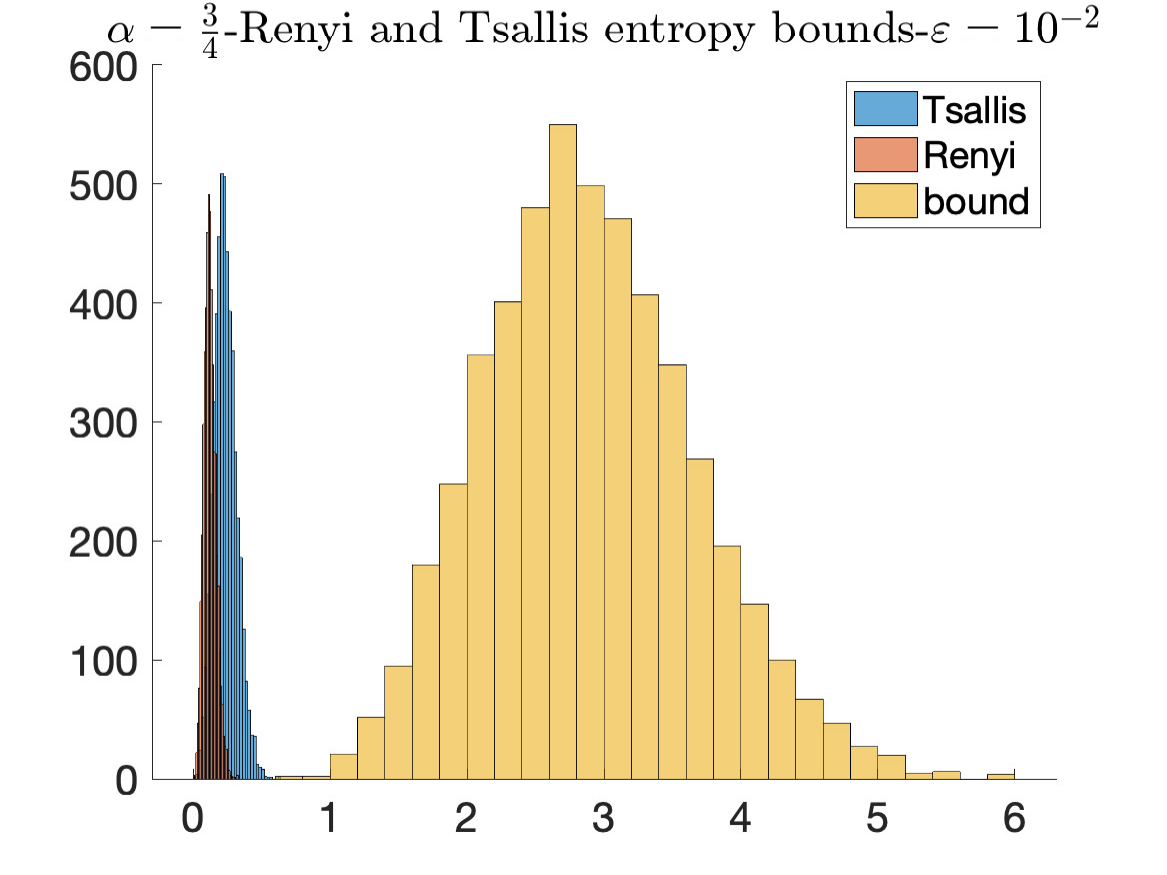}\quad \includegraphics[width=7.5cm]{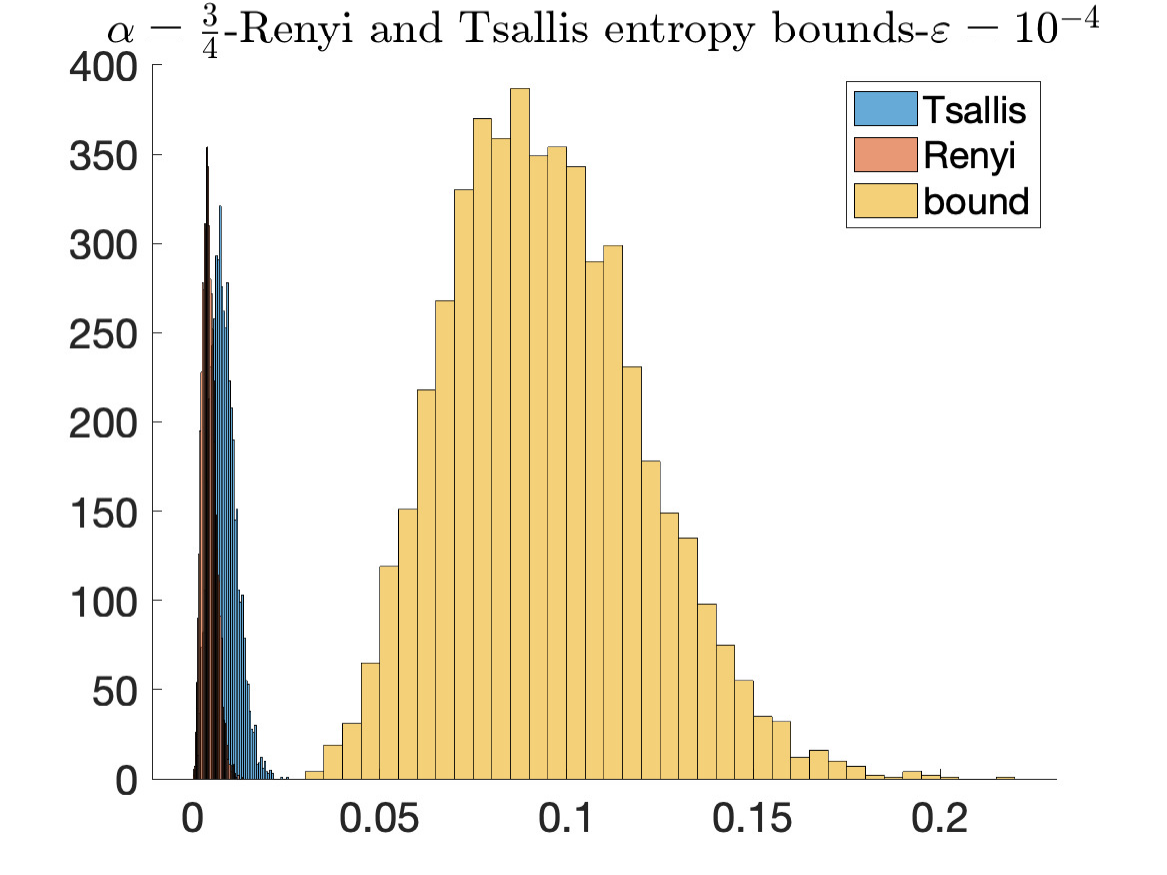}\\[5pt]
\includegraphics[width=7.5cm]{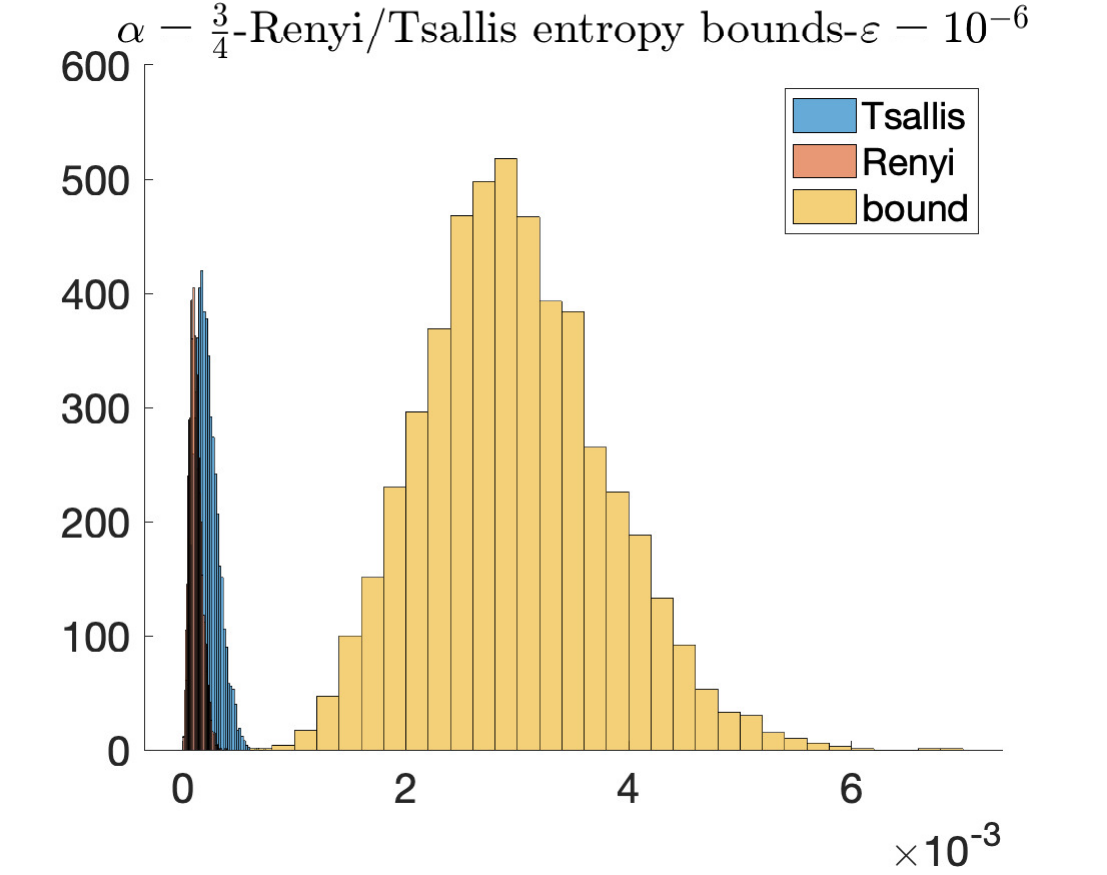}\quad \includegraphics[width=7.5cm]{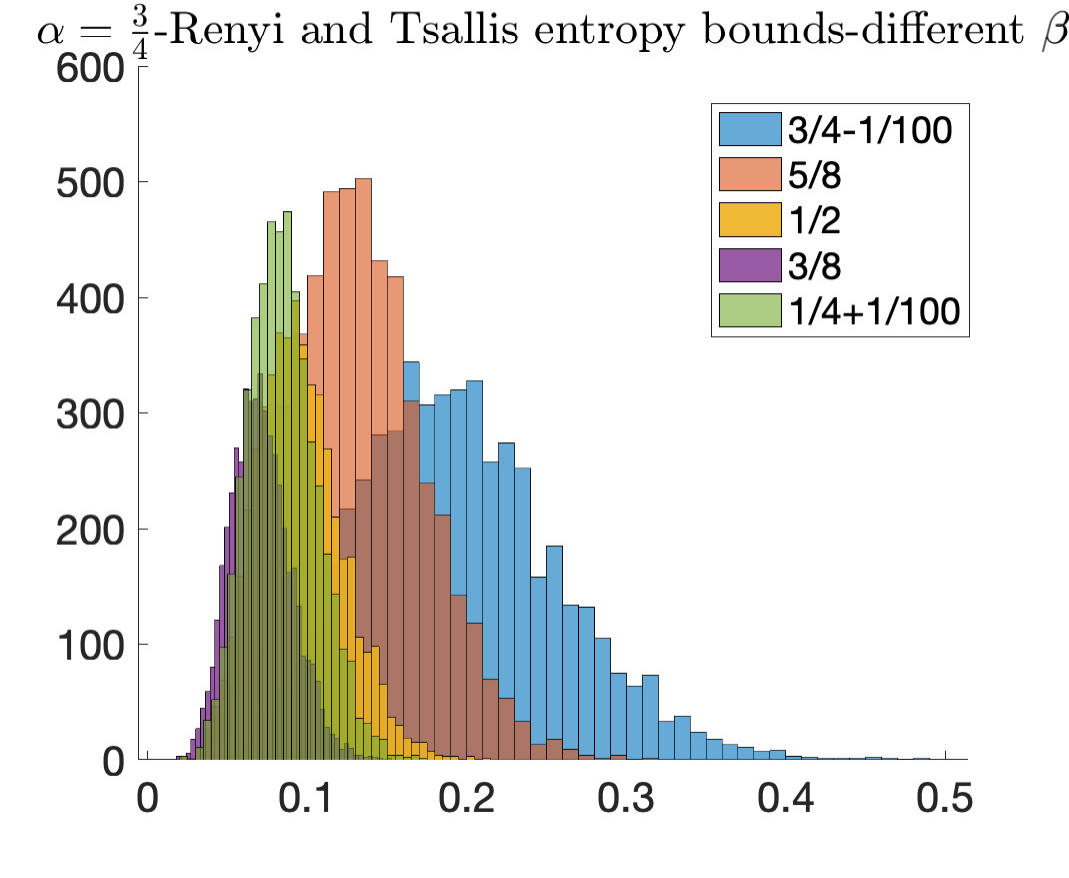}
\end{center}
\caption{These figures illustrate our findings of Corollary \ref{corr:contdisc}. We compute for 5000 random distributions $p$ on $\{1,2,...,1000\}$ the $\alpha$-Tsallis and $\alpha$-R\'enyi entropy by perturbing the distribution by $\varepsilon q$ where $q$ is another random vector. In the fourth plot we compute the bounds for different values of $\beta$ and observe that low values of $\beta$ yield better bounds. We choose the weights $w_i=i.$ In this histogram, the $x$-axis depicts the absolute value of the difference of the $\alpha$-Tsallis and $\alpha$-R\'enyi entropy for the realizations of our 5000 sample distributions and also the value of our bound for these realizations. The $y$-axis shows the number of times this value on the $x$-axis was achieved among the 5000 realizations.}
\end{figure*}

\subsubsection{Continuous random variables}
In a recent paper \cite{POP}, a continuity bound for the differential entropy has been obtained, but analogous bounds for the $\alpha$-R\'enyi and $\alpha$-Tsallis entropies are missing. Therefore, we now turn to the study of continuous random variables $X,Y$ with densities $\mu_X,\mu_Y$ on some domain $D$ and derive continuity bounds for the $\alpha$-R\'enyi and Tsallis entropies. This leads us then to the study of $\alpha$-H\"older continuous functions $f:[0,\max\{\Vert \mu_X \Vert_{\infty},\Vert \mu_Y\Vert_{\infty}\}] \rightarrow \RR$ defining the quantities
\[ (Tf)(X) \coloneqq \int_D f(\mu_X(x))\ dx.\]
Analogously, to the discrete case, $H(X)=Tf_1(X)$ is the differential \emph{Shannon entropy} of the random variable $X$, $R_{\alpha}(X)=\frac{\log(Tf_{\alpha}(X))}{1-\alpha}$ its differential $\alpha$-\emph{R\'enyi entropy}, and $T_{\alpha}(X)=\frac{Tf_{\alpha}(X)-1}{1-\alpha}$ its differential  $\alpha$-\emph{Tsallis entropy}.
A straightforward adaptation of Proposition \ref{prop:simple} yields
\begin{prop}
\label{prop:simple2}
Let $D$ be some domain and $f \in C^{\alpha}([0,\tau])$, $\alpha \in (0,1)$, with $\tau = \operatorname{max}\{\Vert \mu_X \Vert_{L^{\infty}},\Vert \mu_Y \Vert_{L^{\infty}}\} \in (0,\infty]$ where $\mu_X,\mu_Y$ are probability densities on $D$ associated with random variables $X,Y$ respectively. Let $w$ be a positive weight function such that $w^{-\frac{\beta}{1-\alpha}} \in L^1(D)$ for some $\beta < \alpha$ and $w\mu_X, w\mu_Y  \in L^1(D)$. Then we have the following continuity bound:
\eqal{
& \Vert f(\mu_X)-f(\mu_Y) \Vert_{L^1(D)} \\
& \le \vert f \vert_{\Lambda_{\omega_{\alpha}}} \Vert \mu_X-\mu_Y \Vert^{\beta}_{L^1(w)} \Vert \mu_X-\mu_Y \Vert^{\alpha-\beta}_{L^1} \Vert w^{-\frac{\beta}{1-\alpha}}\Vert^{1-\alpha}_{L^1}.
}
\end{prop}
As in the discrete case, we can therefore conclude the following:
\begin{corr}
\label{corr:contcont}
Let $D$ be a domain, $\alpha \in (0,1)$, and $\mu_X,\mu_Y$ be probability densities on $D$ associated with random variables $X,Y$ respectively. Let $w$ be a positive weight function such that $w^{-\frac{\beta}{1-\alpha}} \in L^1(D)$ for some $\beta < \alpha$ and $w\mu_X, w\mu_Y  \in L^1(D).$ In addition, let $(Tf_{\alpha})(X),(Tf_{\alpha})(Y) \ge \delta^{-1},$
for some $\delta>0$. Then we have the following continuity bounds: 
The $\alpha$-Tsallis entropy satisfies 
\eqal{
& \vert T_{\alpha}(X)-T_{\alpha}(Y) \vert \\
& \le \frac{1 }{1-\alpha}\Vert X-Y\Vert^{\beta}_{\operatorname{TV}((w_i))} \Vert X-Y\Vert_{\operatorname{TV}}^{\alpha-\beta} \Vert (w_i^{-\frac{\beta}{1-\alpha}})\Vert^{1-\alpha}_{\ell^1}.
}
The $\alpha$-R\'enyi entropy satisfies 
\eqal{
& \vert R_{\alpha}(X)-R_{\alpha}(Y) \vert \\
& \le \frac{\delta } {1-\alpha}\Vert X-Y\Vert^{\beta}_{\operatorname{TV}((w_i))} \Vert X-Y\Vert_{\operatorname{TV}}^{\alpha-\beta} \Vert (w_i^{-\frac{\beta}{1-\alpha}})\Vert^{1-\alpha}_{\ell^1}.
}
\end{corr}
\begin{proof}
As in the discrete case, we can now use that $\vert \log(x)-\log(y)\vert \le \delta \vert x-y\vert$ for all $x,y\ge \delta^{-1} >0$, for the $\alpha$-R\'enyi entropy, which we apply to the integrals $Tf_{\alpha}(X),Tf_{\alpha}(Y)$ appearing in the definition of $R_{
alpha}$. The result then immediately follows from the triangle inequality and Proposition \ref{prop:simple2}
\end{proof}
\begin{rem}
Continuity estimates for $\alpha>1$ for classical $\alpha$-R\'enyi and Tsallis entropies can be obtained along the lines of the corresponding quantum mechanical result that we state as Proposition \ref{prop:alphag1}.
\end{rem}

\subsection{Quantum R\'enyi and Tsallis entropies}
\label{sec:QEntropies}
Let $\rho$ be a quantum state, \textit{i.e.} a positive trace-class operator on a separable Hilbert space with unit trace. The spectral theorem implies that for any Borel function $f: \RR \to \CC$ we can write, in terms of rank $1$-projections $\pi_k,$ just $f(\rho)  = \sum_k f(\lambda_k) \pi_k.$
Our first theorem of this section, Theorem~\ref{theo:approx_theo}, is a general continuity result for functions of density operators under moment constraints. The moment constraints stated in the theorem follow for the entropies already from energy constraints on the states themselves. We elaborate on this in Lemma \ref{lemm:moment_bounds} in the appendix.

The theorem crucially relies on two results obtained by Aleksandrov and Peller, \cite[Theo. $5.8$]{AP} and \cite[Theo. $7.1$]{AP1}.
The first important result that is crucial for our purposes is a continuity bound for Schatten norms and H\"older continuous functions:
Fix $\alpha \in (0,1)$ and $p \in (1,\infty)$. There exists a universal constant $c>0$ such that for any function $f \in C^{\alpha}(\RR)$ and $A,B$ self-adjoint operators with $A-B \in \mathcal S_p,$ the operator $f(A)-f(B) \in \mathcal S_{p/\alpha}$ and 
\begin{equation}
\label{eq:Schatten}
\Vert f(A)-f(B) \Vert_{p/\alpha} \le c \vert f \vert_{\Lambda_{\omega_\alpha}} \Vert A-B \Vert_{p}^{\alpha}.
\end{equation} 
It is important to observe here that the result does not allow us to take $p/\alpha=1$ immediately. In fact, a bound for $p/\alpha=1$ has been found, but requires a higher level of regularity \cite[Theo. $6.2$]{AP}. Thus, we cannot directly apply the above estimates in trace distance but need to work in weaker Schatten norms, first.

The second result is a continuity bound that is merely in operator norm but for arbitrary moduli of continuity:
Similarly, for every function $\omega^*$ associated to a modulus of continuity $\omega$ as defined in Section \ref{sec:QCM}, there exists some $c>0$ such that for self-adjoint $A$ and $B$ and any $f \in \Lambda_{\omega}$ we have 
\begin{equation}
\label{eq:Op}
 \Vert f(A) - f(B) \Vert \le c \vert f\vert_{\Lambda_{\omega}} \omega^*(\Vert A-B\Vert).
 \end{equation}
Now, we have all the prerequisites to state the approximation theorem:
\begin{theo}[Approximation theorem]
\label{theo:approx_theo}
Let $\rho,\sigma$ be states, $f$ a measurable function and $\ha$ a positive Hamiltonian with compact resolvent such that for some $\beta>0$ we have $\tr(\ha^{\beta}\vert f \vert(\rho)), \tr(\ha^{\beta}\vert f \vert(\sigma)) \le \mu<\infty.$ Let $\varepsilon>0$, $f \in C^{\alpha}$, and take the spectral projection $P  \coloneqq \indic_{[0,\mu/\varepsilon]}(\hat H^{\beta})$. Then, there is $c>0$ such that for all $w/q<1$ and $p$ conjugate to $q$ 
\eqal{
\Vert f(\rho)- f(\sigma) \Vert_1 & \le \sqrt{8\varepsilon}\Big(\sqrt{\Vert f(\rho) \Vert_1} \\
& \quad + \sqrt{\Vert f(\sigma) \Vert_1}\Big)  +c \vert f \vert_{C^{w/q}}\Vert P\Vert_p \Vert \rho-\sigma\Vert_w,
}
and for general $f \in \Lambda_{\omega}$, modulus of continuity $\omega$, and integrated modulus of continuity $\omega^*$
\eqal{
\Vert f(\rho)- f(\sigma) \Vert_1 & \le \sqrt{8\varepsilon}\Big(\sqrt{\Vert f(\rho) \Vert_1} \\
& \quad + \sqrt{\Vert f(\sigma) \Vert_1}\Big)  +c \vert f \vert_{\Lambda_{\omega}}\Vert P\Vert_1 \omega^*(\Vert \rho-\sigma\Vert).
}
\end{theo}
\begin{proof}
We then find as in the proof of the Gentle Measurement Lemma \cite[Lemma $9$]{W99} that for any projection $P$ 
\begin{equation}
\begin{split}
\label{eq:estm}
& \Vert f(\rho)- Pf(\rho) P \Vert_1^2 \\
& \stackrel{\operatorname{(i)}}{\le} \left( \sum_k \vert f(\lambda_k)\vert \Vert \pi_k - P \pi_k P\Vert_1 \right)^2 \\
& \stackrel{\operatorname{(ii)}}{\le} \sum_k\vert f(\lambda_k)\vert \Vert \pi_k - P \pi_k P\Vert_1^2  \Vert f(\rho) \Vert_1 \\
&  \stackrel{\operatorname{(iii)}}{\le}4 \sum_k\vert f(\lambda_k)\vert (1-\tr(\pi_kP\pi_kP))  \Vert f(\rho) \Vert_1 \\
&  \stackrel{\operatorname{(iv)}}{\le} 8 \sum_k \vert f(\lambda_k)\vert (1-\tr(\pi_kP))  \Vert f(\rho) \Vert_1 \\
& \le 8 (\Vert f(\rho) \Vert_1-\tr(\vert f \vert(\rho)P))  \Vert f(\rho) \Vert_1,
\end{split}
\end{equation}
using (i) the spectral decomposition, (ii) the Cauchy-Schwarz inequality, (iii) the Fuchs-van-de Graaf inequality, cf. \eqref{eq:FvdG}, and (iv) $(1-x^2)\le 2(1-x).$
For a general $\alpha$-Hölder continuous function $f$, assuming that $\tr(\ha^{\beta}\vert f\vert(\rho) )\le \mu < \infty$ and $P$ as in the statement, we obtain
\eqal{
\mu & \ge \tr(\ha^{\beta} \vert f \vert(\rho)) \\
& = \tr(\ha^{\beta}  P\vert f \vert(\rho)) + \tr(\ha^{\beta} (1-P)\vert f \vert(\rho)) \\
& \ge \mu/\varepsilon \tr((1- P)\vert f \vert(\rho)).
}
This implies that  $\tr((1- P)\vert f \vert(\rho)) \le \varepsilon \Rightarrow \Vert f(\rho) \Vert_1- \tr(P \vert f \vert(\rho)) \le \varepsilon.$
From this, we get in \eqref{eq:estm} that $\Vert f(\rho)- Pf(\rho) P \Vert_1 \le \sqrt{8\Vert f(\rho) \Vert_1\varepsilon}.$
Putting it all together, we find for any admissible $w/q<1$, using \cite[Theo. $5.8$]{AP}, and $p$ being the Hölder conjugate exponent to $q$
\begin{equation}
\begin{split}
& \Vert f(\rho) -f(\sigma) \Vert_1 \\
& \le \Vert f(\rho) - P f(\rho) P \Vert_1 + \Vert P (f(\rho) -f(\sigma)) P \Vert_1 \\
& \quad + \Vert P f(\sigma) P  -f(\sigma)) \Vert_1 \\
&\le \sqrt{8\varepsilon}\Big(\sqrt{\Vert f(\rho) \Vert_1} + \sqrt{\Vert f(\sigma) \Vert_1}\Big) \\
& \quad + \Vert P \Vert_p \Vert (f(\rho) -f(\sigma))  \Vert_q\\
&\le \sqrt{8\varepsilon}\Big(\sqrt{\Vert f(\rho) \Vert_1} + \sqrt{\Vert f(\sigma) \Vert_1}\Big) \\
& \quad + c\vert f \vert_{C^{w/q}} \Vert P \Vert_p \Vert \rho-\sigma \Vert_w^{w/q}, 
\end{split}
\end{equation}
where we used \eqref{eq:Schatten} in the last step.
For the general case, we thus find, using \eqref{eq:Op},
\eqal{
\Vert f(\rho)-f(\sigma) \Vert & \le c \sqrt{8\varepsilon}\Big(\sqrt{\Vert f(\rho) \Vert_1} + \sqrt{\Vert f(\sigma) \Vert_1}\Big) \\
& \quad + c \vert f \vert_{\Lambda_{\omega}}\Vert P\Vert_1 \omega^*(\Vert \rho-\sigma\Vert).
}
\end{proof}
Let $(Tf)(\rho)  \coloneqq  \tr(f(\rho)),$ then $S(\rho)=Tf_1(\rho)$ is the \emph{von Neumann entropy} of the density operator $\rho$, $R_{\alpha}(\rho)=\frac{\log(Tf_{\alpha}(\rho))}{1-\alpha}$ its $\alpha$-\emph{R\'enyi entropy}, and $T_{\alpha}(\rho)=\frac{Tf_{\alpha}(\rho)-1}{1-\alpha}$ its $\alpha$-\emph{Tsallis entropy}. Using that $Tf_{\alpha}(\rho)\ge 1,$ we thus find the following immediate corollary, since $\vert \log(x)-\log(y)\vert \le \vert x-y\vert$ for $x,y \ge 1.$ 
\begin{corr}
\label{corr:quantcont}
Let $\alpha \in (0,1)$ and $\rho,\sigma$ be states and $\ha$ a positive Hamiltonian with compact resolvent such that for some $\beta>0$ we have $\tr(\ha^{\beta}f_{\alpha}(\rho)), \tr(\ha^{\beta}f_{\alpha}(\sigma)) \le \mu<\infty.$ Let $\varepsilon>0$ and take the spectral projection $P  \coloneqq \indic_{[0,\mu/\varepsilon]}(\hat H^{\beta})$. Then for all $\alpha=w/q<1$ and $p$ conjugate to $q$,
\eqal{
& \vert T_{\alpha}(\rho) -  T_{\alpha}(\sigma)\vert \\
& \le \frac{\sqrt{8\varepsilon}\Big(\sqrt{ \Vert f_{\alpha}(\rho)\Vert_1} + \sqrt{\Vert f_{\alpha}(\sigma)\Vert_1}\Big)  +c \Vert P\Vert_p \Vert \rho-\sigma\Vert_{\alpha q}^{\alpha}}{1-\alpha}
}
and
\eqal{
& \vert R_{\alpha}(\rho) -  R_{\alpha}(\sigma)\vert \\
& \le \frac{\sqrt{8\varepsilon}\Big(\sqrt{ \Vert f_{\alpha}(\rho)\Vert_1} + \sqrt{\Vert f_{\alpha}(\sigma)\Vert_1}\Big)  +c \Vert P\Vert_p \Vert \rho-\sigma\Vert_{\alpha q}^{\alpha}}{1-\alpha}.
}
\end{corr} 
\noindent Estimates on $\Vert f_{\alpha}(\rho)\Vert_1,\Vert f_{\alpha}(\sigma)\Vert_1$ can be found in Lemma \ref{lemm:moment_bounds}.

\subsection{The case $\alpha>1$} 
The case $\alpha>1$ is fundamentally different from the case $\alpha \in (0,1)$ studied before. For completeness, we state the following Proposition on $\alpha$-R\'enyi and $\alpha$-Tsallis entropies for $\alpha>1.$ For $\alpha>1$, the $\alpha$-Tsallis entropy in fact becomes Lipschitz continuous, as has already been observed in \cite{tsallis}. This is different from the $\alpha$-R\'enyi entropy which is not uniformly continuous for any $\alpha>1.$
\begin{prop}
\label{prop:alphag1}
Let $\alpha>1$.
The $\alpha$-Tsallis entropy is always Lipschitz continuous with respect to the $\alpha$-Schatten distance:
$$\vert T_{\alpha}(\rho) -T_{\alpha}(\sigma)\vert = \frac{\vert \Vert \rho \Vert^{\alpha}_{\alpha} - \Vert \sigma \Vert^{\alpha}_{\alpha} \vert}{\alpha-1}\le \frac{\alpha}{\alpha-1} \Vert \rho-\sigma\Vert_{\alpha}.$$
The $\alpha$-R\'enyi entropy is not uniformly continuous on the set of states unless an energy constraint is imposed. Thus, let $\rho,\sigma$ be states such that at least one of the following two conditions holds,
\begin{enumerate}
    \item $\tr(\rho^{\alpha}),\tr(\sigma^{\alpha}) \ge \delta^{-1} >0$ for some $\delta>0$ or 
    \item There exists a positive definite Hamiltonian such that $\tr(\ha \rho),\tr(\ha\sigma) \le E<\infty$ and for some $\beta \in (0,1),$ we have $\tau \coloneqq \tr\Big(\ha^{-\frac{\alpha \beta}{(1-\beta)(\alpha-1)}}\Big) <\infty$ such that we can define $\delta \coloneqq E^{\frac{\alpha\beta}{1-\alpha}} \tau^{\alpha-1}>0,$
\end{enumerate}
then 
\[ \vert R_{\alpha}(\rho) - R_{\alpha}(\sigma) \vert  \le \frac{\alpha \delta}{1-\alpha} \Vert \rho - \sigma \Vert_{\alpha}.\]
\end{prop}
\begin{proof}
For the $\alpha$-R\'enyi entropy, we proceed as follows. Let $\Spec(\rho) = \{ \lambda_n; n \in \mathbb N\}$ and $\Spec(\ha)=\{ \mu_n; n \in \mathbb N\}$, then we find
\begin{equation}
    \begin{split}
    \label{eq:hoelderrenyi}
        1 &= \sum_{i=1}^{\infty} \lambda_i = \sum_{i=1}^{\infty} \mu_i^{\beta} \lambda_i^{\beta} \lambda_i^{1-\beta} \mu_i^{-\beta} \\
        &\stackrel{(1)} \le \left( \sum_{i=1}^{\infty} \mu_i \lambda_i \right)^{\beta} \left( \sum_{i=1}^{\infty} \frac{1}{\mu_i^{\frac{\beta}{1-\beta}}} \lambda_i \right)^{1-\beta} \\
        & \stackrel{(2)} \le (\tr(\ha \rho))^{\beta} (\tr(\rho^{\alpha}))^{\frac{1-\beta}{\alpha}} \left( \sum_{i=1}^{\infty} \mu_i^{-\frac{\alpha \beta}{(1-\beta)(\alpha-1)}} \right)^{\frac{(1-\beta)(\alpha-1)}{\alpha}},
    \end{split}
\end{equation}
where in (1) and (2) we use H\"older's inequality, and for the first term in (2) we use the Courant-Fischer theorem, Proposition~\ref{prop:CFT}. This implies by rearranging \eqref{eq:hoelderrenyi} that 
\begin{align}
     \Vert \rho \Vert^{\alpha}_{\alpha} &\ge \left( (\tr(\ha \rho))^{\frac{\alpha\beta}{1-\beta}} (\tr(\ha^{-\frac{\alpha \beta}{(1-\beta)(\alpha-1)}}))^{\alpha-1} \right)^{-1} \nonumber\\
&\ge \left(E^{\frac{\alpha\beta}{1-\beta}}\Bigl(\tr\Big(\ha^{-\frac{\alpha \beta}{(1-\beta)(\alpha-1)}}\Big)\Bigr)^{\alpha-1} \right)^{-1}=:\frac{1}{\delta}.
\end{align}

Thus, we have proven that $\tr(\rho^{\alpha}) \ge \delta^{-1}>0$ for some $\delta >0.$ 
Hence, for the $\alpha$-R\'enyi entropy we get, 
\begin{equation}
    \begin{split}
    \label{eq:renyiest}
    \vert R_{\alpha}(\rho) - R_{\alpha}(\sigma) \vert 
    &\stackrel{(1)}{\le} \frac{\delta}{\alpha-1} \vert \Vert \rho \Vert^{\alpha}_{\alpha} - \Vert \sigma \Vert^{\alpha}_{\alpha} \vert \\
    &  \stackrel{(2)}\le \alpha \frac{\delta}{\alpha-1} \vert \Vert \rho \Vert_{\alpha}-  \Vert \sigma \Vert_{\alpha}\vert \\
    & \le  \alpha \frac{\delta}{1-\alpha} \Vert \rho - \sigma \Vert_{\alpha},
\end{split}
\end{equation}
where in (1) we use that $\vert \log(x)-\log(y)\vert \le \delta \vert x-y\vert$ for all $x,y \ge \delta^{-1}$ and in (2) we used that $\vert x^{\alpha} - y^{\alpha}\vert \le \alpha \vert x-y \vert$ for $x,y \in [0,1].$ Both estimates are readily verified using the mean value theorem. 

For the $\alpha$-Tsallis entropy we end up directly estimating as in the last part of \eqref{eq:renyiest} $$\vert T_{\alpha}(\rho) -T_{\alpha}(\sigma)\vert = \frac{\vert \Vert \rho \Vert^{\alpha}_{\alpha} - \Vert \sigma \Vert^{\alpha}_{\alpha} \vert}{\alpha-1}\le \frac{\alpha}{\alpha-1} \Vert \rho-\sigma\Vert_{\alpha}.$$
\end{proof}

\subsection{The Finite-dimensional Approximation (FA) property}
\label{sec:FAproperty}
As has been discussed in the papers~\cite{W78,W15} and has also in this article, continuity bounds for states on infinite-dimensional Hilbert spaces often rely on constraints on the energy of states by some Hamiltonian to make the entropy functional continuous. It is now tempting to turn this question around and ask if there always exists a natural Hamiltonian, defining a Gibbs state, for any state of finite entropy. 

There are various functionals that are not continuous with respect to trace distance. For instance, in any arbitrarily small neighbourhood of a state in an infinite-dimensional Hilbert space, there exists a state of infinite entropy. The following condition that excludes many such pathological states, for various quantities of information theoretic interest, was recently proposed by Shirokov \cite{S21,S21b}:
\begin{defi}[FA-property]
A state $\rho = \sum_{k \in \mathbb N} \lambda_k \pi_k$ satisfies the \emph{FA}-property if there is a sequence $(g_k)$ with $g_k \ge 0$ such that 
\begin{align}
&\sum_{k \in \mathbb N} \lambda_k g_k< \infty  \label{eq:one}\\
\text{ and } \quad &\lim_{\beta \downarrow 0} \left( \sum_{k \in \mathbb N} e^{-\beta g_k }\right)^{\beta} =1.\label{eq:two}
\end{align}
\end{defi}
It was observed in \cite[Theo. 1]{S21} that if a state $\rho$ satisfies the FA property, then $\rho$ has finite von Neumann entropy.  In that same article, the question was raised whether any state $\rho$ that has finite von Neumann entropy, necessarily satisfies the FA property. 

In particular, for $\lambda_k \in \Spec(\rho)$ it has been shown that once  $\sum_{k \ge 1} \lambda_k \log(k)^q <\infty$ for some $q>2,$ then $\rho$ satisfies the FA-property. Yet, there clearly exist states that do not fall in this category.
Let $\rho_{\alpha}$ be a state whose eigenvalues $\lambda_{\alpha,k}$ are, up to a normalizing constant
\begin{equation}
    \nu_{\alpha} \coloneqq \sum_{k \ge 2} \frac{1}{k\log(k)^{\alpha}},
\end{equation}
given as $\lambda_{\alpha,k} =\frac{1}{\nu_{\alpha} (k \log(k)^{\alpha})}$ for $\alpha \in (2,3)$.  
Any such $\rho$ has finite entropy. 
This can be seen as follows:
\begin{align}
    S(\rho_{\alpha}) &= H(\{\lambda_{\alpha,k}\}) = - \sum_{k \ge 2} \lambda_{\alpha,k} \log \lambda_{\alpha,k} = \sum_{k \ge 2} a_k,
\end{align}
where 
$$a_k = \frac{\log \left(\nu_{\alpha} k \log (k)^{\alpha}\right)}{\nu_{\alpha} k \log (k)^{\alpha}}.$$
The terms in the sequence $(a_k)$ are all positive and decrease for $k$ large enough. Note that
 $$\sum_{k \ge 2 } 2^k a_{2^k} = a\sum_{k \ge 2} \frac{1}{k^\alpha} + b\sum_{k \ge 2} \frac{1}{k^{\alpha-1}} + c\sum_{k \ge 2}\frac{\log k}{k^{\alpha}} < \infty,$$
 for some constants $a,b,c$. By Cauchy's condensation test this implies that $\sum_{k } a_k$ converges, and hence $S(\rho) < \infty$.
We now want to argue that any such $\rho$ cannot satisfy the FA-property answering the question raised in \cite{S21} in a negative way:

\begin{theo}
\label{theo:FAprop}
Any state $\rho$ with spectrum $\Spec(\rho)=\left\{\lambda_n; n \in \mathbb N\right\}$ such that $\lambda_n \le \nu/(n \log(n)^{3})$ for almost all $n \ge 2$, with normalizing constant $\nu>0,$ does not satisfy the FA-property.
In particular, the set of states satisfying the FA-property is strictly smaller than the set of finite entropy states.
\end{theo}
\begin{proof}
The comparison test for the sequence implies that for a positive sequence $\lambda_n$ satisfying the condition of Theorem~\ref{theo:FAprop}, $\sum_{k \ge 2} \lambda_k g_k <\infty$ implies
\eqal{\label{eq:351}
    & \nu \sum_{k \ge 4} \frac{g_k}{k \log(k)^2(\log(k)-1)} \\
    & = \nu \sum_{k \ge 4} \frac{g_k}{k \log(k)^3} + \nu \sum_{k \ge 4} \frac{g_k}{k \log(k)^3(\log(k)-1)}< \infty,
}
where the last sum converges by the majorant criterion.
Continuity of the logarithm, implies that the Gibbs part of the condition of the FA-property can be rewritten as
\begin{equation}
\label{eq:Shirokov}
 \lim_{\beta \downarrow 0} \beta \log\left(\sum_{k \in \mathbb N} e^{-\beta g_k } \right) =0.
 \end{equation}
 Our aim is now to show that for any sequence $(g_k)$ such that the first part, \eqref{eq:one}, of the FA-property holds, \textit{i.e.} $\sum_k \lambda_k g_k <\infty$, the second one, \eqref{eq:two} in the form of \eqref{eq:Shirokov}, is necessarily violated.
Writing $A_+ \coloneqq \{k ; g_k \le \log(k)^2\}$ and $A_-$ for its complement, we find by monotonicity and $\mathbb N=A_+\dot\cup A_-$ in (1) and again monotonicity in (2) to estimate the series by an integral
\eqal{ \label{eq:A+}
\sum_{ k \in A_+ } e^{-\beta g_k}
& \stackrel{(1)}{\ge} \sum_{k \in \mathbb N} e^{-\beta \log(k)^2} - \sum_{k \in A_-} e^{-\beta \log(k)^2} \\
& \stackrel{(2)}{\ge}\int_1^{\infty} e^{-\beta \log(x)^2} \ dx - \sum_{k \in A_-} e^{-\beta \log(k)^2}.
} 
Quite explicitly, we observe that by substituting $x=e^u$, we have
\eqal{
\int_1^{\infty} e^{-\beta \log(x)^2} \ dx & =\int_0^{\infty} e^{u(1-\beta u)} \ du \\
& = e^{\frac{1}{4\beta}}\tfrac{\sqrt{\pi} \left (1+\operatorname{erf}\left(\tfrac{1}{2\sqrt{\beta}}\right) \right)}{2\sqrt{\beta}}.
}
Monotonicity of the logarithm yields that for all $k \in A_-$ such that $\log(k) \in [n,n+1)$ we have $e^{-\beta \log(k)^2} \le e^{-\beta n^2}$  
\begin{equation}
\label{eq:logbound}
\sum_{k \in A_-} e^{-\beta \log(k)^2} \le \sum_{n \in \mathbb N} \vert \log(A_-) \cap [n,n+1)\vert e^{-\beta n^2},
\end{equation} 
where $\log(S):=\{\log(s); s\in S\}$ for a suitable set $S.$
We then introduce renormalized coefficients $\alpha_k \coloneqq K(k)/e^k \le (e-1)$, with $K(k) \coloneqq \vert \log(A_-) \cap [k,k+1)\vert,$ where we used for the last inequality the simple worst case estimate by replacing $A_-$ in the definition of $K(k)$ by $\mathbb N$, in which case $K(k)$ may be replaced by $e^{k+1}-e^k$, as they obey the same asymptotic scaling. This way, we find that for any $\delta>0$
\eqal{
\sum_{k \in A_-} e^{-\beta \log(k)^2} & \stackrel{\eqref{eq:logbound}}{\le}  \sum_{k \in \mathbb N} \alpha_k e^{k(1-\beta k)} \\
& \stackrel{(1)}{\le} \delta \int_0^{\infty} e^{u(1-\beta u)} \ du \\
& \quad + \delta e^{\frac{1}{4\beta}}+  \sum_{k \in \mathbb N;\alpha_k \ge \delta} \alpha_k  e^{k(1-\beta k)},
}
where in (1) we split the sum into parts $\{\alpha_k \ge \delta\}$ and its complement and estimated the latter series by an integral and its maximum value.
We then estimate in terms of the cardinality $\gamma_{\beta} \coloneqq \left\vert k \in [\tfrac{1-\sqrt{\beta}}{2\beta},\tfrac{1+\sqrt{\beta}}{2\beta}] \cap \mathbb N; \alpha_k \ge \delta\right\vert,$ by just using monotonicity and $\vert \alpha_k \vert \le e-1$ 
\eqal{
 & \sum_{k \in \mathbb N;\alpha_k \ge \delta} \alpha_k  e^{k(1-\beta k)} \\
 & \le (e-1)\left(\left( \int_0^{\tfrac{1-\sqrt{\beta}}{2\beta}}+\int_{\tfrac{1+\sqrt{\beta}}{2\beta}}^{\infty}\right) e^{x(1-\beta x)} \ dx +\gamma_{\beta} e^{\frac{1}{4\beta}} \right) \\
 & = e^{\frac{1}{4\beta}} \frac{\sqrt{\pi}}{2\sqrt{\beta}} \big( (e-1)(1-2\operatorname{erf}(\tfrac{1}{2}) + \operatorname{erf}(\tfrac{1}{2\sqrt{\beta}})) \\
 & \quad \quad \quad \quad \quad \quad +\tfrac{2(e-1)}{\sqrt{\pi}}\sqrt{\beta} \gamma_{\beta} \big),
} 
where we used in addition that by differentiation we see that $e^{k(1-\beta k)} \le e^{1/(4\beta)}$ for $k \in [\tfrac{1-\sqrt{\beta}}{2\beta},\tfrac{1+\sqrt{\beta}}{2\beta}].$
Applying all the above inequalities to \eqref{eq:A+}, we find for $\delta>0$ sufficiently small, which we shall assume from now on,

\begin{align}
    \sum_{ k \in {\mathbb{N}} } e^{-\beta g_k}&\geq \sum_{ k \in A_+ } e^{-\beta g_k}\geq e^{\frac{1}{4\beta}}\left[A(\beta)(1-\delta) - \delta - B(\beta) \right],
\end{align}
where
\begin{equation}
A(\beta) \coloneqq \tfrac{\sqrt{\pi}(1+\operatorname{erf}((2\sqrt{\beta})^{-1}))}{2\sqrt{\beta}}    
\end{equation}
and
\eqal{
B(\beta) & \coloneqq \frac{\sqrt{\pi}}{2\sqrt{\beta}} \Big( (e-1)(1-2\operatorname{erf}(\tfrac{1}{2}) + \operatorname{erf}(\tfrac{1}{2\sqrt{\beta}})) \\
& \quad \quad \quad \quad \quad +\tfrac{2(e-1)}{\sqrt{\pi}} \beta^{1/2}\gamma_{\beta} \Big).
}
Recall that we want to show $\lim_{\beta \downarrow 0} \beta \log\left(\sum_{k \in \mathbb N} e^{-\beta g_k } \right)\neq 0$.
So upon taking the logarithm and multiplying by $\beta$ the exponential function $e^{\frac{1}{4\beta}}$ contributes a term $\beta \log(e^{\frac{1}{4\beta}})=\frac{1}{4}$ which is non-zero. To show that \eqref{eq:Shirokov} does not hold, it suffices therefore to show that, for fixed $\delta>0$ small,  $A(\beta)(1-\delta)-\delta - B(\beta)$ is bounded uniformly away from zero for choices of $\beta=\beta_n$ with a sequence $(\beta_n)$ that tends to zero. In fact, we have that
\begin{equation}
    \label{eq:goal}
    A(\beta)(1-\delta)-\delta - B(\beta) =\frac{1}{\sqrt{\beta}}( C(\beta) - (e-1) \beta^{1/2}\gamma_\beta),
\end{equation}
with
\eqal{
    C(\beta) & = \frac{\sqrt{\pi}}{2 }\Bigg[ (1+\operatorname{erf}(\tfrac{1}{2\sqrt{\beta}}))(1-\delta)-\frac{2\sqrt{\beta}}{\sqrt{\pi}} \delta \\
    & \quad - \left((e-1)(1-2\operatorname{erf}(\tfrac{1}{2}) + \operatorname{erf}(\tfrac{1}{2\sqrt{\beta}})) \right) \Bigg].
}
Since with our standing assumption that $\delta>0$ is sufficiently small, one verifies that $C(\beta) \ge \varepsilon >0$ for some sufficiently small $\varepsilon>0$ and all $\beta $ small enough, it suffices therefore, in order to show that \eqref{eq:goal} is uniformly positive, that $\beta^{1/2}\gamma_{\beta}$ is strictly smaller than $\frac{\varepsilon}{2(e-1)}$ for a sequence of $\beta_n$ tending to zero. In fact, this implies that $C(\beta_n)-(e-1)\beta_n^{1/2}\gamma_{\beta_n}\ge c >0$ for some $c>0$ and all $n \in \mathbb N$ which appears in the right hand side of \eqref{eq:goal}.  
Thus by applying Lemma \ref{lemm:auxlemm} to $f(t)=t^{1/2} \gamma_t$, it suffices to show the finiteness of the following integral, where we without loss of generality restrict ourselves to $k \ge 4$,

\begin{equation*}
\begin{split}
\int_0^{1} \gamma_{\beta} \frac{d\beta}{\beta^{1/2}} &\stackrel{(1)}{=} 2 \int_{1}^{\infty} \gamma_{\frac{1}{\alpha^2}} \frac{d\alpha}{\alpha^2} \\
& \stackrel{(2)}{=}2 \sum_{k; \alpha_k \ge \delta} \int_1^{\infty} \indic(k)_{\Big[\tfrac{\alpha(\alpha-1)}{2},\tfrac{\alpha(\alpha+1)}{2}\Big]} \ \frac{d\alpha}{\alpha^2}\\
&\stackrel{(3)}{=} 2  \sum_{k; \alpha_k \ge \delta}  \int_1^{\infty} \indic(\alpha)_{\Big[\tfrac{\sqrt{8k+1}-1}{2},\tfrac{\sqrt{8k+1}+1}{2}\Big]} \frac{d\alpha}{\alpha^2} \\
& \stackrel{(4)}{=}  \sum_{k; \alpha_k \ge \delta} \frac{1}{k} \\
&\stackrel{(5)}\le \sum_{k} \frac{\alpha_k}{\delta k} \\
& \stackrel{(6)}= \sum_{k} \frac{K(k)}{\delta e^kk} \\
& \stackrel{(7)}= \sum_{k} \sum_{\substack{r \in A_-:\\\atop{\log(r) \in \log(A_-) \cap [k,k+1)}}}\frac{e}{\delta e^{k+1}k}\\
&\stackrel{(8)}\le \sum_{k} \sum_{\substack{r \in A_-:\\\atop{\log(r) \in \log(A_-) \cap [k,k+1)}}}\frac{e}{\delta r(\log(r)-1)} \\
& \stackrel{(9)}= \sum_{r \in A_-} \frac{e}{\delta r(\log(r)-1)}\\
&\stackrel{(10)}\le \sum_{r \in A_-} \frac{eg_r}{\delta r\log(r)^2(\log(r)-1)} <\infty,
\end{split}
\end{equation*}
where we 
\begin{enumerate}
    \item substituted $\beta=\alpha^{-2}$,
    \item  used the definition of $\gamma_\beta$,
    \item used an equivalent representation of the indicator function,
    \item evaluated the integral,
    \item used that for the $k$ we are summing over $\alpha_k/\delta>1$,
    \item used the definition of $\alpha_k,$
    \item partitioned the summation according to the definition of $K(k)$,
    \item used that $r(\log(r)-1) \le  e^{k+1}k$ on the respective partitions $\{\log(r) \in \log(A_-) \cap [k,k+1)\}$,
    \item dropped the partitioning over $k$,
    \item used the definition of $A_-$ and that the final sum has to be finite by the first condition of the FA property, \eqref{eq:one} which in our case is \eqref{eq:351}.
\end{enumerate}
\end{proof}

\section{Conclusion and Open questions}

\noindent \textbf{Conclusion}

In this article, we provide for the first time a tight continuity estimate for the classical Shannon entropy for random variables on $\mathbb N_0$ and von Neumann entropy of quantum states on a separable, infinite-dimensional Hilbert space whose energy is constrained by the number operator. 

We also provide for the first time continuity estimates for $\alpha$-R\'enyi and  $\alpha$-Tsallis entropies for $\alpha \in (0,1)$ both in classical probability theory and in quantum mechanics for infinite-dimensional Hilbert spaces. By doing so, we provide a tool to derive general continuity bounds for Hölder continuous functions of density operators. 

Finally, we show that the finiteness of the von Neumann entropy of a state does not imply the FA-property.~\\

\noindent \textbf{Future directions}

As possible future directions, we would like to propose the following list of open problems:

\begin{enumerate}
    \item Provide a tight continuity estimate for the Shannon entropy of random variables on general countable alphabets.
    \item Generalize the tight continuity bounds for the von Neumann entropy to general Hamiltonians satisfying the Gibbs hypothesis, see Def. \ref{defi:Gibbs}. Our method of proof seems to apply to this more general framework as well, but requires the solution of an optimization problem.
    \item Derive continuity estimates for different moments. Due to its fundamental relevance in the uncertainty principle, it seems also reasonable to request a bound on the variance of the energy or other appropriate quantum observables instead of the energy.
    \item Provide a tight continuity estimate for the differential entropy of random variables with densities.
    \item We give sufficient criteria for the finiteness of $\alpha$-R\'enyi and $\alpha$-Tsallis entropies for states satisfying certain energy constraints. Do there exist also necessary criteria?
    \item Investigate the tightness of the continuity bounds for $\alpha$-R\'enyi and $\alpha$-Tsallis entropies.
    \item The set of states of finite von Neumann entropy is strictly larger than the set of states satisfying the FA property. Can one obtain analogous results to the ones obtained in \cite{S21,S21b} only assuming the finiteness of the entropy?
    \item In \cite{W15} similar continuity estimates as for the von Neumann entropy have also been obtained for the conditional von Neumann entropy, cf. Lemma 17. Can one provide a tight version of that Lemma too?
    \item As in the preceding open problem, similar questions about tightness can be asked for estimates on quantum conditional mutual information, the Holevo quantity and for capacities of quantum channels, as investigated in \cite{Shirokovchannel}.
\end{enumerate}

\section*{Acknowledgments}
The authors are grateful to Maksim Shirokov for his helpful comments and for pointing out that Theorems~\ref{theo:FanoG},~\ref{theo:ContinuityG}, and~\ref{theo:ContinuityQuantumG} require an additional constraint. ND would also like to thank Yury Polyanskiy for valuable feedback.
SB~gratefully acknowledges support by
the UK Engineering and Physical Sciences Research Council (EPSRC)
grant EP/L016516/1 for the University of Cambridge Centre for Doctoral
Training, the Cambridge Centre for Analysis.
MGJ~gratefully acknowledges support from the Carlsberg Foundation under Grant CF19-0313.

\appendix

\section{Moment bounds}
We now state some moment bounds on $f_{\alpha}$ and $f_1$ using energy constraints on the state. The proofs give rise to slightly sharper bounds, that are less concise to state, than the ones we outline in the statement of the Lemma. Therefore, the reader might want to consult the proof of the following Lemma for slightly improved estimates. For notational simplicity, we use the notation $(\tr(A))^k$ to denote $\tr(A)^k$ below.
\begin{lemm}
\label{lemm:moment_bounds}
Let $\ha$ be a positive Hamiltonian with compact resolvent and $\rho$ a state. Then, as soon as the right-hand side is finite, we have the following moment bounds for the function $f_1(x)=-x\log(x)$
\begin{equation*}
\tr(\ha^{1/2} f_1(\rho)) \le \tr(\ha\rho) + \tr\left(e^{-\sqrt{\ha}}(1+\ha)\right). 
\end{equation*}
Similarly, for $f_{\alpha}(x) =x^{\alpha}$, we have
\begin{equation}
    \tr\Big( \ha^{1/2}f_{\alpha}(\rho)\Big) \le \tr(\ha\rho)+ \tr\left(\ha^{\frac{2\alpha-1}{2(\alpha-1)}}\right)
\end{equation}
for $\alpha \in (1/2,1)$, and
\begin{equation}
    \tr(\ha^{\beta} f_{\alpha}(\rho)) \le \tr(\ha\rho)^{\alpha} \tr(\ha^{-\frac{r}{1-\alpha}})
\end{equation}
for $\beta=\alpha-r$ and $r \in (0,\alpha]$.
Furthermore, in terms of the number operator, we find $\tr(\ha^{-\frac{r}{1-\alpha}})<\infty$ for $\ha =\hat N+1$ and $\alpha>1/2$ as well as $\ha =  (\hat N+1)^{\kappa}$ for $\kappa<\frac{1-\alpha}{r}$, with $r=\alpha-\varepsilon$ for $1/2 \ge \alpha >\varepsilon>0.$ 
\end{lemm}
\begin{proof}
For the following computations, we recall the H\"older inequality $\tr(\ha^{\beta}\pi_k) \le \tr(\ha\pi_k)^{\beta}$ for $\beta \in (0,1)$. Then writing $I \coloneqq \{k;1 \le \tr(\ha^{1/2}\pi_k)\}$, $I^c \coloneqq \mathbb{N} \setminus I$, $J \coloneqq \{k ;\log( \lambda_k) > -\tr(\ha^{1/2}\pi_k)\}$ and $J^c \coloneqq \mathbb{N} \setminus J$, we have

\begin{equation}
\begin{split}
 &\tr(\ha^{1/2} f_1(\rho))= -\sum_{k \in \mathbb N} \lambda_k \log(\lambda_k) \tr(\pi_k \ha^{1/2} ) \\
 & \stackrel{(1)}{\le} \sum_{k \in J} \lambda_k\tr(\pi_k \ha^{1/2} )^2 \\
 & +  \sum_{k \in J^c} e^{-\tr(\ha^{1/2}\pi_k) } \Big( \tr(\pi_k \ha^{1/2} )^2 \indic_I(k) \\
 & \quad \quad \quad \quad \quad \quad \quad \quad \quad \quad + \tr(\pi_k \ha^{1/2})    \indic_{I^c} \Big) \\
  & \stackrel{(2)}{\le} \tr(\ha\rho) + \sum_{k \in J^c} e^{-\tr(\ha^{1/2}\pi_k) } \Big( \tr(\pi_k \ha^{1/2} )^2 \indic_I(k) \\
  & \quad \quad \quad \quad \quad \quad \quad \quad \quad \quad + \tr(\pi_k \ha^{1/2}) \indic_{I^c}(k) \Big),
 \end{split}
\end{equation}

where (1) follows by separately estimating different eigenvalues and using that $x^2 \ge x$ for $x \ge 1$ and (2) follows from H\"older's inequality. 
We find using a similar splitting for $f_{\alpha}(x)=x^{\alpha},$ 
\begin{equation}
\begin{split}
 & \tr(\ha^{1/2} f_{\alpha}(\rho)) \\
 &\le \sum_{k\in \mathbb N} \lambda_k \lambda_k^{\alpha-1} \tr(\ha\pi_k)^{1/2} \\
 &\stackrel{(1)}\le  \sum_{\substack{k\\ \lambda_k^{\alpha-1}<\tr(\ha\pi_k)^{1/2}}} \lambda_k \tr(\ha\pi_k)  \\
 & \quad +  \sum_{\substack{k \\\lambda_k^{\alpha-1}>\tr(\ha\pi_k)^{1/2}}} \tr(\ha\pi_k)^{\frac{1}{2}(1-\frac{\alpha}{1-\alpha})}\\
&\le \tr(\ha\rho) +\sum_{\substack{k \\\lambda_k^{\alpha-1}>\tr(\ha\pi_k)^{1/2}}} \frac{1}{ \tr(\ha\pi_k)^{\frac{2\alpha-1}{2(1-\alpha)}}},
 \end{split}
\end{equation}
where in (1) we use that $\lambda^{\alpha-1}\le \tr(\hat{H}\pi_k)$ in the first term and $\lambda^{\alpha} \le \tr(\hat \pi_k)^{-\frac{\alpha}{2(1-\alpha)}}$ in the second term.
Aside from special cases, this bound is only saleable for $\alpha>1/2.$ To satisfactorily treat also the cases $\alpha \in (0,1)$, by choosing $p=1/\alpha$ and its H\"older conjugate $q=1/(1-\alpha)$, we have by H\"older's inequality that for $\beta =\alpha-r$ and any $r \in (0,\alpha)$
\begin{equation}
\begin{split}
& \tr(\ha^{\beta} f_{\alpha}(\rho)) \\
& = \sum_{k \in \mathbb N} \lambda_k^{\alpha } \tr(\ha^{\beta}\pi_k)\tr(\ha\pi_k)^{r} \tr(\ha\pi_k)^{-r} \\
 &\le \left(\sum_{k \in \mathbb N} \lambda_k \tr(\ha\pi_k)^{\frac{\beta+r}{\alpha}}  \right)^{\alpha} \left(\sum_{k \in \mathbb N} \tr(\ha\pi_k)^{-\frac{r}{1-\alpha}} \right)^{1-\alpha} \\
 &= \tr(\ha\rho)^{\alpha} \left(\sum_{k \in \mathbb N} \tr(\ha\pi_k)^{-\frac{r}{1-\alpha}} \right)^{1-\alpha}.
 \end{split}
\end{equation}
Finally, we have using in (1) that $e^{-x}\operatorname{max}\{x^2,x\}\le e^{-x}(x^2+1)$ and in (2) that $e^{-x}(x^2+1)$ is monotonically decreasing together with the Courant-Fischer theorem stated in Proposition~\ref{prop:CFT}, that for the eigenbasis of the Hamiltonian with rank-1 projections $p_k$ corresponding to the ordered eigenvalues $\lambda_1 \le \lambda_2\le...$ of the Hamiltonian
\begin{equation}
\begin{split}
&\sum_{k \in J^c} e^{-\tr(\ha^{1/2}\pi_k) } \Big( \tr(\pi_k \ha^{1/2} )^2  \indic_{I}(k) \\
& \quad \quad \quad \quad \quad \quad \quad \quad \quad \quad + \tr(\pi_k \ha^{1/2})  \indic_{I^c}(k) \Big) \\
&\stackrel{(1)}\le \sum_{k \in J^c} e^{-\tr(\ha^{1/2}\pi_k) } \left(1+ \tr(\pi_k \ha^{1/2} )^2  \right) \\
&\stackrel{(2)}\le  \sum_{k \in J^c} e^{-\tr(\ha^{1/2}p_k) } \left(1+ \tr( \ha^{1/2}p_k )^2  \right) \\
& \stackrel{(3)}\le\Tr(e^{-\sqrt{\ha}}(\hat{H}+1)),
\end{split}
\end{equation}
where we dropped the constraint on $k$ in (3).

Analogously, we have again used monotonicity and the Courant Fischer theorem
\begin{equation}
\begin{split}
&\sum_{k} \frac{1}{ \tr(\ha\pi_k)^{\frac{2\alpha-1}{2(1-\alpha)}} }  \le \tr\left(\ha^{-\frac{2\alpha-1}{2(1-\alpha)}}\right), \\
& \sum_{k \in \mathbb N} \tr(\ha\pi_k)^{-\frac{r}{1-\alpha}} \le \tr\left(\ha^{-\frac{r}{1-\alpha}}\right).
\end{split}
\end{equation}
\end{proof}

\bibliographystyle{IEEEtran}
\bibliography{BiblioInfoTheory}

\begin{IEEEbiographynophoto}{Simon Becker}
	is a mathematical physicist who has completed his undergraduate studies in mathematics and physics at the Free University of Berlin and Ludwig Maximilian University in Munich. After obtaining his master's degree, he obtained a PhD in applied mathematics from the University of Cambridge in 2021.

Simon spent a year at the Courant Institute at New York University, following the completion of his PhD, where he continued his research in mathematical aspects of condensed matter physics. Currently, he is a postdoctoral researcher in the mathematics department at ETH Zurich.
\end{IEEEbiographynophoto}

\begin{IEEEbiographynophoto}{Nilanjana Datta}
	obtained a PhD in mathematical physics  from ETH Zurich, Switzerland, in 1996.
From 1997 to 2000, she was a Postdoctoral Researcher at the Dublin
Institute of Advanced Studies, C.N.R.S. Marseille, and EPFL in
Lausanne. In 2001 she joined the University of Cambridge, U.K., as a
Lecturer in Mathematics of Pembroke College, and a member of the
Statistical Laboratory in the Centre for Mathematical Sciences. She is
currently a Professor of Quantum Information Theory in the Department of Applied Mathematics
and Theoretical Physics of the  University of Cambridge, and a Fellow of Pembroke College. Her
scientific interests include quantum information theory and mathematical physics.
\end{IEEEbiographynophoto}

\begin{IEEEbiographynophoto}{Michael G. Jabbour}
	received the B.S. and M.S. degrees in physics engineering from \'Ecole polytechnique de Bruxelles, Universit\'e libre de Bruxelles, Brussels, in 2013 and the Ph.D. degree in engineering sciences from \'Ecole polytechnique de Bruxelles, Universit\'e libre de Bruxelles, Brussels, in 2018.

	From 2019 to 2020, he was a Postdoctoral Researcher in the Department of Applied Mathematics and Theoretical Physics, in the Centre for Mathematical Sciences of the University of Cambridge. Since 2021, he has been a Postdoctoral Researcher in the Department of Physics of the Technical University of Denmark. His research interests include quantum optics, quantum information theory and mathematical physics.
\end{IEEEbiographynophoto}

\end{document}